\documentclass[fleqn]{LMCS}

\def\dOi{10(1:14)2014}
\lmcsheading%
{\dOi}
{1--29}
{}
{}
{Jun.~\phantom05, 2013}
{Jun.~24, 2014}
{}

\ACMCCS{[{\bf Software and its engineering}]: Software organization
  and properties---Software functional properties---Formal methods;
  [{\bf Theory of computation}]: Formal languages and automata
  theory---Automata extensions---Quantitative automata }


\allowdisplaybreaks

\usepackage{microtype}

\usepackage{amsmath}
\usepackage{amsfonts}
\usepackage{amssymb}
\usepackage{amsthm}

\usepackage{algorithmic, algorithm}

\usepackage{tikz}
\usetikzlibrary{decorations.pathreplacing}
\usetikzlibrary{decorations.pathmorphing}
\usetikzlibrary{shapes,shapes.symbols,automata,arrows}
\usetikzlibrary{calc}

\tikzset{p0/.style = {shape = circle,    draw, thick, minimum size = 0.4cm}}
\tikzset{p1/.style = {rectangle split, rectangle split parts=2, rectangle split draw splits = false, minimum width=.8cm, draw, thick}}
\tikzset{>=stealth, shorten >=1pt}
\tikzset{every edge/.style = {thick, ->, draw}}
\tikzset{every loop/.style = {thick, ->, draw}}

\usepackage{booktabs}

\newcommand{\set}[1]{\{ #1 \}}
\newcommand{\nats}{\mathbb{N}}

\renewcommand{\epsilon}{\varepsilon}
\newcommand{\last}{\mathrm{Lst}}

\newcommand{\arena}{\mathcal{A}}
\newcommand{\game}{\mathcal{G}}
\newcommand{\cost}{\mathrm{Cst}}
\newcommand{\costfam}{{\overline{\cost}}}
\newcommand{\col}{\Omega}
\newcommand{\win}{\mathrm{Win}}

\newcommand{\parity}{\mathrm{Parity}}
\newcommand{\cp}{\mathrm{CostParity}}
\newcommand{\bcp}{\mathrm{BndCostParity}}

\newcommand{\buechi}{\mathrm{B\ddot{u}chi}}
\newcommand{\coB}{\mathrm{coB\ddot{u}chi}}

\newcommand{\streett}{\mathrm{Streett}}
\newcommand{\cs}{\mathrm{CostStreett}}
\newcommand{\bcs}{\mathrm{BndCostStreett}}

\newcommand{\rr}{\mathrm{RR}}
\newcommand{\scrr}{\mathrm{SCRR}}
\newcommand{\pcrr}{\mathrm{PCRR}}
\newcommand{\mem}{\mathcal{M}}
\newcommand{\init}{\mathrm{Init}}
\newcommand{\update}{\mathrm{Upd}}
\newcommand{\nxt}{\mathrm{Nxt}}

\newcommand{\paritydist}{\mathrm{Cor}}
\newcommand{\streettdist}{\mathrm{StCor}}

\newcommand{\lexord}{\le}
\newcommand{\lexordstrict}{<}

\newcommand{\sheet}{\mathrm{Sh}}

\newcommand{\sub}{\mathrm{sub}}
\newcommand{\answer}[1]{\mathrm{Ans}({#1})}

\newcommand{\req}{\mathrm{Req}}
\newcommand{\creq}{\mathrm{ReqCst}}
\newcommand{\score}{\mathrm{Sc}}

\newcommand{\att}[3]{\mathrm{Attr}_{#1}^{#2}(#3)}

\newcommand{\up}{\mathbf{UP}}
\newcommand{\coup}{\mathbf{coUP}}
\newcommand{\np}{\mathbf{NP}}
\newcommand{\conp}{\mathbf{coNP}}
\newcommand{\exptime}{\mathbf{EXPTIME}}
\newcommand{\ptime}{\mathbf{PTIME}}
\newcommand{\pspace}{\mathbf{PSPACE}}

\usepackage{hyperref}

\begin{document}

\title[Parity and Streett Games with Costs]{Parity and Streett Games with 
Costs\rsuper*}

\author[N.\ Fijalkow]{Nathana\"{e}l Fijalkow\rsuper a}
\address{{\lsuper a}LIAFA, Universit{\'e} Paris 7 \and 
		 Institute of Informatics, University of Warsaw}
\email{nath@liafa.univ-paris-diderot.fr}

\author[M.\ Zimmermann]{Martin Zimmermann\rsuper b}
\address{{\lsuper b}Reactive Systems Group, Saarland University}
\email{zimmermann@react.uni-saarland.de}



\titlecomment{{\lsuper*}A preliminary version of this work appeared in FSTTCS 2012
under the name 
``Cost-parity and Cost-Streett Games''~\cite{FijalkowZimmermann12}.
The research leading to these results has received funding from the European 
Union's Seventh Framework Programme (FP7/2007-2013) under grant agreements 
259454 (GALE) and 239850 (SOSNA)}


\begin{abstract}
We consider two-player games played on finite graphs equipped with costs on
edges and introduce two winning conditions, cost-parity and cost-Streett,
which require bounds on the cost between requests and their responses. Both
conditions generalize the corresponding classical omega-regular conditions and
the corresponding finitary conditions.

For parity games with costs we show that the first player has
positional winning strategies and that determining the winner lies in
NP and coNP. For Streett games with costs we show that the first
player has finite-state winning strategies and that determining the
winner is EXPTIME-complete. The second player might need infinite
memory in both games. Both types of games with costs can be solved by
solving linearly many instances of their classical variants.
\end{abstract}

\maketitle

\section{Introduction}
In recent years, boundedness problems arose in topics pertaining to automata
and logics leading to the development of novel models and techniques to tackle
these problems. Although in general undecidable, many boundedness problems for
automata turn out to be decidable if the acceptance condition can refer to
boundedness properties of variables, but the transitions cannot access
variable values. A great achievement was made by
Hashiguchi~\cite{DBLP:journals/jcss/Hashiguchi82} who proved decidability of
the star-height problem by reducing it to a boundedness problem for a certain
type of finite automaton and by then solving this problem. This led the path
to recent developments towards a general theory of bounds in automata and
logics, comprising automata and logics with
bounds~\cite{Bojanczyk04,BojanczykColcombet06}, satisfiability algorithms for
these logics~\cite{Bojanczyk11,BojanczykTorunczyk12, Boom11}, and regular
cost-functions~\cite{Colcombet09}.

In this work, we consider boundedness problems in turn-based two-player graph
games of infinite duration. We introduce cost-parity and cost-Streett
conditions which generalize the (classical) $\omega$-regular parity-
respectively Streett condition, as well as the finitary parity- respectively
finitary Streett condition~\cite{ChatterjeeHenzingerHorn09}. 

A game with cost-parity condition is played on an arena whose vertices are 
colored by natural numbers, and where traversing an edge incurs a non-negative cost.
The cost of a play (prefix) is the sum of the costs of the edges along the play (prefix). Player~$0$ wins a play if there is a bound $b$ such that all but finitely many
odd colors seen along the play (which we think of as requests) are followed by
a larger even color (which we think of as responses) that is reached with cost
at most $b$. 
If all edges have cost zero, then this reduces to the parity condition: all but finitely many
odd colors are followed by a larger even color.
If all edges have positive cost, then this reduces to finitary parity conditions:
there is a bound $b$ such that all but finitely many odd colors
are followed by a larger even color within $b$ steps. The definition of the cost-Streett condition goes along the same
lines, but the requests and responses are independent and not hierarchically
ordered as in parity conditions. 

The cost of traversing an edge can be used to
model the consumption of a resource. Thus, if Player~$0$ wins a play she can
achieve her goal along an infinite run with bounded resources. On the other
hand, Player~$1$'s objective is to exhaust the resource, no matter how big the
capacity is. Note that this is not an $\omega$-regular property, which is
witnessed by the fact that Player~$1$ needs infinite memory to win such games.

Since the term ``cost-parity games'' has been
used before~\cite{Colcombet09,CL10,Boom11}, we refer to games with cost-parity
conditions as parity games with costs. 
The first difference between cost-parity games and parity games with costs is the bound quantification:
in cost-parity games the counter values are required to be uniformly bounded over all paths, 
whereas in parity games with costs the bound can depend on the path. 
However, as shown in~\cite{CF13} in a more general context, the two formulations
are equivalent over finite arenas.
The actually relevant difference between cost-parity games and parity games with costs is in the intent: 
cost-parity games were introduced to solve the domination
problem for regular cost-functions over finite trees~\cite{CL10}.
Hence cost-parity games are very general: their winning conditions are conjunctions of a parity condition
and of boundedness requirements on counters, 
allowing the counters and the parity condition to evolve independently.
In contrast to this work, we are interested in efficient algorithms to solve games.
Hence, in parity games with costs we restrict the use of counters, 
which are only used to give a quantitative measure of the satisfaction of the
parity condition.
This gives rise to a better-behaved winning condition, for which we provide a finer analysis
of the complexity of solving the games and of the memory requirements.

We show that parity games with costs enjoy two nice properties of parity and
finitary parity games: Player~$0$ has memoryless winning strategies and
determining the winner lies in $\np \cap \conp$~\footnote{This was recently improved to $\up \cap \coup$~\cite{DBLP:conf/lpar/MogaveroMS13}.}. Furthermore, we show that
solving parity games with costs can be algorithmically reduced to solving
parity games, which allows to solve these games almost as efficiently as
parity games. We then consider Streett games with costs and prove that
Player~$0$ has finite-state winning strategies, and that determining the
winner is $\exptime$-complete. 

This unifies the previous results about finitary parity and Streett games and the results about their
classical variants, in the following sense. For both parity and Streett,
recall that the games with costs generalize both the classical and the finitary variants,
hence solving them is at least as hard as solving these two subcases.
Our results show that it does not get worse: solving games with costs is not harder 
than solving the corresponding classical and finitary games.
Indeed, solving finitary parity games can be carried
out in polynomial time~\cite{ChatterjeeHenzingerHorn09}, while no
polynomial-time algorithm for parity games is yet known, and the decision
problem for parity games is in $\np \cap \conp$.
The situation is reversed for Streett games, 
since solving them is $\conp$-complete~\cite{EmersonJutla99}
while solving finitary Streett games is $\exptime$-complete. 
The latter result is shown in unpublished work by Chatterjee, Henzinger, and Horn: by slightly
modifying the proof of $\exptime$-hardness of solving request-response games
presented in~\cite{ChatterjeeHenzingerHorn11} they prove $\exptime$-hardness
of solving finitary Streett games.

To obtain our results, we present an algorithm to solve parity games with
costs that iteratively computes the winning region of Player~$0$ employing an
algorithm to solve parity games. This ``reduction'' to parity games also
yields finite-state winning strategies for Player~$0$ in parity games with
costs. However, this can be improved: by exploiting the intrinsic structure of
the memory introduced in the reduction, we are able to prove the existence of
positional winning strategies for Player~$0$. We also give a second proof of
this result: we show how to transform an arbitrary finite-state winning
strategy into a positional one. This construction relies on so-called scoring
functions (which are reminiscent of the simulation of alternating
tree-automata by non-deterministic automata presented in~\cite{MullerSchupp95}
and of scoring functions for Muller games~\cite{McNaughton00}) and presents a
general framework to turn finite-state strategies into positional ones, which
we believe to be applicable in other situations as well. Finally, we present
an algorithm that solves Streett games with costs by solving Streett games.
Here, we show the existence of finite-state winning strategies for Player~$0$
in Streett games with costs.

Adding quantitative requirements to qualitative winning conditions has been an
active field of research during the last decade: much attention is being paid
to not just synthesize some winning strategy, but to find an optimal one
according to a certain quality measure, e.g., the use of mean-payoff
objectives and weighted automata to model quantitative aspects in the winning
condition~\cite{BCHJ09, CernyChatterjeeHenzingerRadhakrishnaSingh11,
ChatterjeeHenzingerJurdzinski05}. For request-response games and their
extensions, waiting times between requests and their responses are used to
measure the quality of a strategy and it was shown how to compute optimal
(w.r.t.\ the limit superior of the mean waiting time) winning
strategies~\cite{HornThomasWallmeier08, Zimmermann09}. However, the optimal
finite-state strategies that are obtained are exponentially larger than the
ones computed by the classical algorithm.

Finally, there has been a lot of interest in so-called energy games, whose
winning conditions ask for the existence of an initial amount of energy such that a positive energy level is maintained throughout the play. Solving energy games with multiple resources is in general
intractable~\cite{FahrenbergJuhlLarsenSrba11} while so-called consumption
games, a subclass of energy games, are shown to be tractable in~\cite{BCKN12}.
Furthermore, energy parity games, whose winning conditions are a conjunction
of a (single resource) energy and a parity condition, can be solved in $\np
\cap \conp$ and one player (the spoiling one) has positional winning
strategies while the other one needs exponential
memory~\cite{ChatterjeeDoyen10}. The memory requirements show that energy parity games are incomparable to parity games with costs, since the second player needs infinite memory in the latter one. This also implies that there are no continuous reductions between these games via finite memory structures (see the next section for a formal definition of such reductions).

The paper is organized as follows. In Section~\ref{section_defs}, we define
the necessary material related to games and introduce cost-parity and
cost-Streett conditions, as well as their bounded variants, which are used to solve games with costs. In
Section~\ref{section_solvingbcp}, we study bounded parity games with costs,
providing an algorithm to solve them and tight memory requirements for winning
strategies. In Section~\ref{section_solvingcp}, we show how to reduce the
problem of solving parity games with costs to the problem of solving bounded
parity games with costs. In Section~\ref{section_elimmem}, we give a different
proof of the existence of positional strategies for (bounded) parity games
with costs, via scoring-functions. In Section~\ref{section_streett}, we study
Streett games with costs.

\section{Definitions}
\label{section_defs}

We denote the non-negative integers by $\nats$ and define $[n] = \set{0, 1,
\ldots, n-1}$ for every $n \ge 1$.

An \textit{arena}~$\arena=(V, V_0, V_1, E)$ consists of a finite, directed graph~$(V,
E)$ and a partition~$\{V_0, V_1\}$ of $V$ into the positions of Player~$0$
(drawn as circles) and the positions of Player~$1$ (drawn as rectangles). A
\textit{play} in $\arena$ starting in $v\in V$ is an infinite path~$\rho = \rho_0
\rho_1 \rho_2 \cdots$ through $(V, E)$ such that $\rho_0 = v$. 
To avoid the nuisance of dealing with finite plays, 
we assume every vertex to have an outgoing edge.

A \textit{game} $\game = ( \arena, \win )$ consists of an arena $\arena$ and a set
$\win \subseteq V^\omega$ of winning plays for Player~$0$. The set of winning
plays for Player~$1$ is $V^\omega \setminus \win$. We say that $\win$ is
\textit{prefix-independent}, if $\rho \in \win$ if and only if $w\rho \in \win$ for
every play prefix~$w$ and every infinite play~$\rho$.

A \textit{strategy} for Player~$i$ is a mapping $\sigma \colon V^*V_i \rightarrow V$
such that $(v, \sigma(wv)) \in E$ for all $wv \in V^* V_i$. We say that
$\sigma$ is \textit{positional} if $\sigma(wv) = \sigma(v)$ for every $wv \in V^*V_i$.
We often view positional strategies as a mapping~$\sigma \colon V_i \rightarrow V$. 
A play $\rho_0 \rho_1 \rho_2 \ldots$ is \textit{consistent} with $\sigma$ if
$\rho_{n+1} = \sigma( \rho_0 \cdots \rho_n)$ for every~$n$ with $\rho_n \in
V_i$. A strategy~$\sigma$ for Player~$i$ is a \textit{winning strategy} from a set of
vertices $W \subseteq V$ if every play that starts in some $v \in W$ and is
consistent with $\sigma$ is won by Player~$i$. 
The \textit{winning region}~$W_i(\game)$ of Player~$i$ in $\game$ is the set of vertices 
from which Player~$i$ has a winning strategy. 
We say that a strategy is \textit{uniform}, if it is winning from all $v \in W_i(\game)$. 
We always have $W_0(\game) \cap W_1(\game) = \emptyset$. 
On the other hand, if~$W_0(\game) \cup W_1(\game) = V$, then we say that $\game$ is \textit{determined}.
All games we consider in this work are determined. 
\textit{Solving} a game amounts to determining its winning regions and winning strategies.

A \textit{memory structure}~$\mem = (M, \init, \update)$ for an arena $(V, V_0, V_1,
E)$ consists of a finite set~$M$ of memory states, an initialization
function~$\init \colon V \rightarrow M$, and an update function~$\update
\colon M \times V \rightarrow M$. The update function can be extended to
$\update^+ \colon V^+ \rightarrow M$ in the usual way: $\update^+(\rho_0) =
\init(\rho_0)$ and $\update^+ (\rho_0 \cdots \rho_n \rho_{n+1}) =
\update(\update^+(\rho_0 \cdots \rho_n), \rho_{n+1})$. A next-move function
(for Player~$i$) $\nxt \colon V_i \times M \rightarrow V$ has to satisfy $(v,
\nxt(v, m)) \in E$ for all $v \in V_i$ and all $m \in M$. It induces a
strategy~$\sigma$ for Player~$i$ with memory~$\mem$ via
$\sigma(\rho_0\cdots\rho_n) = \nxt(\rho_n, \update^+(\rho_0 \cdots \rho_n))$.
A strategy is called \textit{finite-state} if it can be implemented by a memory
structure.

An arena $\arena = (V, V_0, V_1, E)$ and a memory structure $\mem = (M, \init,
\update)$ for $\arena$ induce the expanded arena $\arena\times\mem = (V \times
M, V_0 \times M, V_1 \times M, E' )$ where $((v,m), (v',m')) \in E'$ if and
only if $(v,v') \in E$ and $\update(m, v' ) = m'$. Every play $\rho$ in
$\arena$ has a unique extended play $\rho' = (\rho_0, m_0) (\rho_1, m_1)
(\rho_2, m_2) \ldots$ in $\arena \times \mem$ defined by $m_0 = \init( \rho_0
)$ and $m_{n+1} = \update(m_n, \rho_{n+1})$, i.e., $m_n = \update^+(\rho_0
\cdots \rho_n)$.

A game $\game = (\arena, \win)$ is \textit{reducible} to $\game' = (\arena', \win')$
via $\mem$, written $\game \le_{ \mem } \game'$, if $\arena' = \arena \times
\mem$ and every play $\rho$ in $\game$ is won by the player who wins the
extended play $\rho'$ in $\game'$, i.e., $\rho \in \win$ if and only if $\rho'
\in \win'$.

\begin{lem} 
\label{lemma_reductionlemma} 
Let $\game$ be a game with vertex set $V$ and $W \subseteq V$. If $\game \le_{
\mem } \game'$ and Player~$i$ has a positional winning strategy for $\game'$
from $\{(v, \init(v))\mid v \in W\}$, then she has a finite-state winning
strategy for $\game$ from $W$ which is implemented by $\mem$.
\end{lem}

Especially, if Player~$i$ has a uniform positional winning strategy for
$\game'$, then she has a uniform finite-state winning strategy for $\game$
that is implemented by $\mem$.

Let $\arena = (V, V_0, V_1, E)$ and $i \in \set{0,1}$. The $i$-\textit{attractor} of
$F\subseteq V$ in $\arena$, denoted by $\att{i}{\arena}{F}$, is defined by
$\att{i}{\arena}{F} = \bigcup_{j=0}^{|V|} A_j$, where $A_0 = F$ and
\begin{align*}
A_{j+1}=A_j\,
\cup\,& \set{v \in V_i     \mid \exists v'\in A_j \text{ such that }(v,v') 
\in E}\\
\cup\,& \set{v \in V_{1-i} \mid \forall v', (v,v') \in E \text{ implies }
v' \in A_j}\enspace.
\end{align*}
Player~$i$ has a positional strategy such that every play that starts in
$\att{i}{\arena}{F}$ and is consistent with the strategy visits $F$. Such
strategies are called \textit{attractor strategies}.

A \textit{trap} for Player~$i$ is a set~$X$ of vertices such that the successors of
every vertex in $X \cap V_i$ are again in $X$ and every vertex in $X \cap
V_{1-i}$ has a successor in $X$. Player~$1-i$ has a positional strategy such
that every play that starts in a trap~$X$ and is consistent with the strategy
stays in $X$ forever. The complement of an attractor~$\att{i}{\arena}{F}$ is a
trap for Player~$i$. Furthermore, removing an attractor from an arena never
introduces terminal vertices.

The following observation will be useful later: if the set of winning plays
$\win$ in $\game$ is prefix-independent, then we have $W_i(\game) =
\att{i}{\arena}{W_i(\game)}$ and $W_i(\game)$ is a trap for Player~$1-i$.
Furthermore, no play consistent with a winning strategy for Player~$i$ will
ever leave $W_i(\game)$.

\subsection{Winning Conditions}
\label{subsection_winningconds}

In this subsection, we present the winning conditions we consider in this
paper. Fix an arena~$\arena$ with set of edges~$E$. A cost function for
$\arena$ is an edge-labelling~$\cost \colon E \rightarrow \{\epsilon,i\}$. Edges labelled with $i$ are called increment-edges while edges labelled by
$\epsilon$ are called $\epsilon$-edges accordingly. We extend the
edge-labelling to a cost function over plays obtained by
counting the number of increment-edges traversed during the play, i.e., $\cost(\rho) \in \nats \cup \set{\infty}$. The cost of a play prefix is defined analogously.

Note that our definition of a cost function only allows cost zero or one on an edge. Alternatively, one could allow arbitrary costs in $\nats$. This would not change our results, as we are interested in boundedness questions only. For the sake of simplicity, we refrain from using arbitrary costs in $\nats$ and use abstract costs~$\epsilon$ and $i$ instead.

\subsubsection{Cost-Parity Conditions}
\label{subsubsec_winningconds_costparity}

Let $\arena = (V, V_0, V_1, E)$ be an arena and let $\col \colon V \rightarrow
\nats$ be a coloring of its vertices by natural numbers. In all games we are
about to define in this subsection, we interpret the occurrence of a color as
request, which has to be answered by visiting a vertex of larger or equal even color at
an equal or later position. By imposing conditions on the responses we obtain several
different types of winning conditions. To simplify our notations, let
$\answer{c} = \set{c' \in \nats \mid c' \ge c \text{ and $c'$ is even}}$ be the
set of colors that answer a request of color~$c$. Note that $\answer{c}
\subseteq \answer{c'}$ for $c \ge c'$ and $c \in \answer{c}$ if $c$ is even.

Fix a cost function $\cost$ and consider a play~$\rho = \rho_0 \rho_1 \rho_2
\cdots$ and a position~$k\in\nats$. We define the cost-of-response at 
position~$k$ of $\rho$ by 
\begin{equation*}
\paritydist_\cost(\rho, k) = \min \set{ \cost (\rho_k \cdots \rho_{k'}) \mid 
k' \ge k \text{ and } \col(\rho_{k'}) \in \answer{\col(\rho_k)} }\enspace, 
\end{equation*}
where we use $\min \emptyset = \infty$, i.e., $\paritydist_\cost(\rho, k)$ 
is the cost of the infix of $\rho$ from position~$k$ to its first 
answer, and $\infty$ if there is no answer.

We say that a request at position~$k$
is answered with cost~$b$, if $\paritydist_\cost(\rho, k) = b$. Note that a
request at a position~$k$ with an even color is answered with cost zero.
Furthermore, we say that a request at position~$k$ is unanswered with
cost~$\infty$, if there is no position~$k' \ge k$ such that $\col(\rho_{k'})
\in \answer{\col(\rho_k)}$ and we have $\cost(\rho_k\rho_{k+1}\cdots ) =
\infty$, i.e., there are infinitely many increment-edges after position~$k$,
but no answer. Note that there is a third alternative: a request can be unanswered with finite cost, i.e., in case it is not answered, but the play~$\rho$ contains only finitely many increment-edges.

We begin defining winning conditions by introducing the parity condition,
denoted by $\parity(\col)$, which requires that all but finitely many requests
are answered. Equivalently, $\rho \in \parity(\col)$ if and only if the
maximal color that occurs infinitely often in $\rho$ is even. Both players
have uniform positional winning strategies in parity
games~\cite{Emersonjutla91, Mostowski91} and their winning regions can be
decided in $\np \cap \conp$ (and even in $\up \cap \coup$\footnote{A problem is in $\up$, if it can be decided by a non-deterministic polynomial-time Turing machine with at most one accepting run.}~\cite{Jurdzinski98}).

By bounding the costs between requests and their responses, we strengthen the
parity condition and obtain the cost-parity and the bounded cost-parity
condition. The former is defined as
\[\cp(\col, \cost) = \set{ \rho \in V^\omega \mid \limsup_{k\rightarrow 
\infty} \paritydist_\cost(\rho , k) < \infty } \enspace,\]
i.e., $\rho$ satisfies the cost-parity condition, if there exists a bound~$b
\in \nats$ such that all but finitely many requests are answered with cost
less than $b$. The bounded cost-parity condition, denoted by $\bcp(\col, \cost)$, is again
obtained by a strengthening: 
\begin{align*}
\bcp(\col, \cost) = \set{ \rho \in V^\omega \mid& 
\limsup_{k\rightarrow  \infty} \paritydist_\cost(\rho , k) < \infty \text{ and }\\
&\text{no request in $\rho$ is unanswered with cost $\infty$}
} \enspace,\end{align*}
i.e., $\rho$ satisfies the bounded cost-parity condition, if there exists a
bound~$b \in \nats$ such that all but finitely many requests are answered with
cost less than $b$, and there is no unanswered request of
cost~$\infty$. Note that this is \emph{not} equivalent to requiring that there
exists a bound~$b' \in \nats$ such that all requests are answered with cost
less than $b'$ (e.g., if there are unanswered requests in a play with
finitely many increment-edges).

\begin{rem}
\label{rem_winningconds_inclusions_parity}
We have $\bcp(\col, \cost) \subseteq \cp(\col, \cost) \subseteq \parity(\col)$ 
and $V^* \cdot \bcp(\col, \cost) = \cp(\col, \cost)$. Furthermore, 
$\cp(\col, \cost)$ and $\parity(\col)$ are prefix-independent while 
$\bcp(\col, \cost)$ is not. 
\end{rem}

A game~$\game = (\arena, \cp(\col, \cost))$ is called a parity game with
costs, and a game with winning condition~$\bcp(\col,\cost)$ is a bounded 
parity game with costs. 
Note that both cost-conditions defined here generalize the
classical parity conditions as well as the finitary, respectively bounded,
parity conditions of~\cite{ChatterjeeHenzingerHorn09}. 
Indeed, if $\arena$ contains no increment-edges, then $\cp(\col, \cost) = \bcp(\col, \cost) = \parity(\col)$,
the three conditions are equivalent. 
On the other hand, if $\arena$ contains no $\epsilon$-edges, 
then $\cp(\col, \cost)$ is equal to the
finitary parity condition over $\col$ and $\bcp(\col, \cost)$ is equal to the
bounded parity condition over $\col$. Hence, parity games with costs
generalize both parity and finitary parity games. Similarly, bounded parity
games with costs generalize both parity and bounded parity games.

Since (bounded) cost-parity conditions are Borel, we obtain determinacy of (bounded) parity games with costs via the Borel determinacy theorem.

\begin{lem}
\label{lemma_costparitydeterminacy}
(Bounded) parity games with costs are determined.
\end{lem}

\begin{proof}
We show that both conditions are Borel. Then, the result follows from the Borel determinacy theorem~\cite{Martin75}.

We have 
\[ \cp(\col, \cost) = \bigcup_{\substack{n \in \nats\\ b\in \nats}} \bigcap_{k \ge n} \bigcup_{k' \ge k} L_{b,k,k'} \]
where $L_{b,k,k'}$ is the set of plays where the request at position~$k$ is answered at position~$k'$ with cost at most $b$. Every $L_{b,k,k'}$ is closed, hence $\cp(\col, \cost)$ is in level~$\Sigma_4$ of the Borel hierarchy\footnote{Note that this is not optimal: the cost-parity condition can be encoded in weak MSO with the unbounding quantifier. Such languages are boolean combinations of $\Sigma_2$ languages~\cite{Bojanczyk11}, which is optimal.}.

Furthermore, $\bcp(\col, \cost)$ is equal to the intersection of $\cp(\col, \cost)$ and the set
\[\set{ \rho \in V^\omega \mid \text{no request in $\rho$ is unanswered with cost $\infty$} }\enspace.\]
The latter set is recognizable by a parity automaton and therefore Borel. Hence, closure of Borel sets under intersection yields the desired result.
\end{proof}

\begin{exa}
\label{example_costparity}
Consider the parity game with costs depicted in Figure~\ref{figure_arena}
where all vertices belong to $V_1$, and the label of a vertex denotes its name
(in the upper part) and its color (in the lower part). Player~$1$ wins from
$\set{a,b,c}$ by requesting color~$1$ at vertex~$a$ infinitely often and
staying at vertex~$b$ longer and longer, but also visiting $c$ infinitely
often (and thereby answering the request). Note that this strategy is not
finite-state. Indeed, one can easily prove that Player~$1$ does not have a
finite-state winning strategy for this game. Player~$0$ wins from every other
vertex, since Player~$1$ can raise only finitely many requests from these
vertices, albeit these requests are unanswered with cost $\infty$.

If we consider the game as a bounded parity game with costs, then Player~$1$
wins from every vertex but $g$ by moving to $g$ and then staying there ad
infinitum. Every such play contains a request of color~$1$ that is unanswered
with cost~$\infty$. From $g$, Player~$0$ wins, since there is only one play
starting from $g$, in which no request is ever raised.

\begin{figure}
\centering
\begin{tikzpicture}
\node[p1] at (0,0)              (a) {$a$ \nodepart{second} $1$};
\node[p1] at (1.5,0)    (b) {$b$ \nodepart{second} $0$};
\node[p1] at (3,0)              (c) {$c$ \nodepart{second} $2$};
\node[p1] at (4.5,0)    (d) {$d$ \nodepart{second} $1$};
\node[p1] at (6,0)              (e) {$e$ \nodepart{second} $0$};
\node[p1] at (7.5,0)    (f) {$f$ \nodepart{second} $1$};
\node[p1] at (9,0)              (g) {$g$ \nodepart{second} $0$};

\path
(a) edge node[below] {$\epsilon$} (b)
(b) edge[loop below]  node[below] {$i$} ()
(b) edge node[below] {$\epsilon$} (c)
(c.north) edge[bend right = 20] node[above] {$\epsilon$} (a.north)
(c) edge node[below] {$\epsilon$} (d)
(d) edge node[below] {$\epsilon$} (e)
(e) edge[loop below]  node[below] {$i$} ()
(e) edge node[below] {$\epsilon$} (f)
(f) edge node[below] {$\epsilon$} (g)
(g) edge[loop below]  node[below] {$i$} ();
\end{tikzpicture}
\caption{A (bounded) parity game with costs.}
\label{figure_arena}
\end{figure}
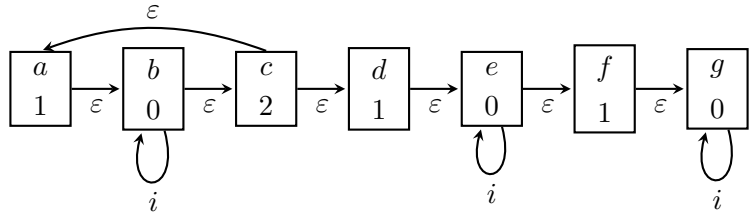
\end{exa}

\subsubsection{Cost-Streett Conditions}
\label{subsubsec_winningconds_coststrett}

Fix an arena~$\arena = (V, V_0, V_1, E)$. Let $\Gamma = (Q_c, P_c)_{c \in
[d]}$ be a collection of $d$ (Streett) pairs of subsets of $V$, i.e., $Q_c,
P_c \subseteq V$, and let $\costfam = (\cost_c)_{c \in [d]}$ be a collection
of $d$ cost functions for $\arena$. We think of visits to vertices in $Q_c$ as
requests, visits to $P_c$ as responses, and measure the cost of these
responses using $\cost_c$. Formally, for $c \in [d]$, a play $\rho = \rho_0
\rho_1 \rho_2 \cdots$, and a position~$k$ we define the cost-of-response by
\begin{equation*}
\streettdist_{\cost_c}(\rho, k)=\begin{cases}
0                                       & \text{if $\rho_k \notin Q_c$,}\\
\min\set{ \cost_c(\rho_k\cdots \rho_{k'}) \mid k' \ge k \text{ and } 
\rho_{k'} \in P_c} & \text{if $\rho_k \in Q_c$,}
\end{cases}
\end{equation*}
where we use $\min \emptyset = \infty$. We define $\streettdist_\costfam(\rho,
k) = \max\set{ \streettdist_{\cost_c}(\rho, k) \mid c\in [d] }$ and say that
the requests at position $k$ are answered with cost $b$, if
$\streettdist_\costfam(\rho, k) = b$, and that the requests are unanswered
with cost $\infty$, if $\streettdist_\costfam(\rho, k) = \infty$ and there are
infinitely many increment-edges after position $k$ (w.r.t.\ some $\cost_c$
such that $\rho_k \in Q_c$). This rules out the case where we have $\streettdist_\costfam(\rho, k) = \infty$ due to a request at position~$k$ that is not answered, but $\rho$ only traverses finitely many increment-edges. 

We consider the following winning conditions. The (classical) Streett
condition $\streett(\Gamma)$ requires for every $c$ that $P_c$ is visited
infinitely often if $Q_c$ is visited infinitely often, i.e., all but finitely
many requests are answered.

Again, by requiring a bound on the costs between requests and responses, we
strengthen the Streett condition: the cost-Streett condition 
\[\cs(\Gamma, \costfam) = \set{ \rho \in V^\omega \mid \limsup_{
k\rightarrow \infty} \streettdist_\costfam(\rho, k) < \infty }
\enspace\] 
requires the existence of a bound~$b$ such that all but finitely many requests
are answered with cost less than $b$. Finally, the bounded cost-Streett condition $\bcs(\Gamma, \costfam)$ requires
the existence of a bound~$b$ such that all but finitely many requests are
answered with cost less than $b$, and that there is no unanswered
request of cost~$\infty$. Formally, we define
\begin{align*}
\bcs(\Gamma, \costfam) = \set{ \rho \in V^\omega \mid& 
\limsup_{k\rightarrow \infty} \streettdist_\costfam(\rho, k) < \infty \text{ and }\\
&\text{no request in $\rho$ is unanswered with cost $\infty$}
} \enspace.\end{align*}

\begin{rem}
\label{rem_winningconds_inclusions_streett}
We have $ \bcs(\Gamma, \costfam) \subseteq \cs(\Gamma, \costfam) \subseteq
\streett(\Gamma)$ and $V^* \cdot \bcs(\Gamma, \costfam) = \cs(\Gamma, \costfam)$.
Furthermore, 
$\cs(\Gamma, \cost)$ and $\streett(\Gamma)$ are prefix-independent while 
$\bcs(\Gamma, \cost)$ is not.
\end{rem}

A game~$(\arena, \cs(\Gamma, \costfam))$ where $\Gamma$ and $\costfam$ have
the same size is called a Streett game with costs. 
As in the case for (bounded) cost-parity conditions,
the winning conditions defined here generalize the classical Streett condition
as well as the finitary, respectively, bounded Streett condition
of~\cite{ChatterjeeHenzingerHorn09}. Indeed, if $\arena$ contains no
increment-edges, then the three conditions are equivalent, i.e., $\cs(\Gamma,
\costfam) = \bcs(\Gamma, \costfam) = \streett(\Gamma)$. Similarly, if $\arena$
contains no $\epsilon$-edges, then $\cs(\Gamma, \costfam)$ is equal to the
finitary Streett condition over $\Gamma$ and $\bcp(\Gamma, \costfam)$ is equal
to the bounded Streett condition over~$\Gamma$. Hence, Streett games with
costs generalize both Streett and finitary Streett games. Similarly, bounded
Streett games with costs generalize both Streett and bounded Streett games.
Furthermore, just as classical Streett games subsume parity games, Streett
games with costs subsume parity games with costs, and bounded Streett games 
with costs subsume bounded parity games with costs.

Figure~\ref{fig_expressiveness} shows the expressiveness of the winning
conditions, e.g., the arrow from ``bounded parity'' to ``bounded Streett''
denotes that every bounded parity condition is also a bounded Streett
condition.

\begin{figure}
\begin{center}
\begin{tikzpicture}[node distance = 3cm]
\node at (0,  0)		(p) 	{parity};
\node at (0,  2.8)		(s)		{Streett};

\node at (-2.8, 1.5) 	(bcp)  	{bounded cost-parity};
\node at (-5.6, 0) 		(bp)  	{bounded parity};

\node at (2.8, 1.5)		(cp)	{cost-parity};
\node at (5.6, 0)		(fp)	{finitary parity};

\node at (-2.8, 4.3)	(bcs)	{bounded cost-Streett};
\node at (-5.6,  2.8)	(bs)	{bounded Streett};

\node at (2.8, 4.3)		(cs)	{cost-Streett};
\node at (5.6,  2.8)	(fs)	{finitary Streett};

\path
(p) edge (s)
(p) edge (cp)
(p) edge (bcp)
(cp) edge (cs)
(bcp) edge (bcs)
(fp) edge (cp)
(bp) edge (bcp)
(fs) edge (cs)
(bs) edge (bcs)
(s) edge (cs)
(s) edge (bcs)
(bp) edge (bs)
(fp) edge (fs);

\path[draw,dashed, rounded corners = .5cm, thick]
(-7,-.4) -- (-1,-.4) -- (-1,3.2) -- (0, 3.7) -- (1,3.2) -- (1,-.4) -- 
(7, -.4);

\end{tikzpicture}
\end{center}
\caption{Expressiveness of winning conditions; those below the dashed line are $\omega$-regular.}
\label{fig_expressiveness}
\end{figure}
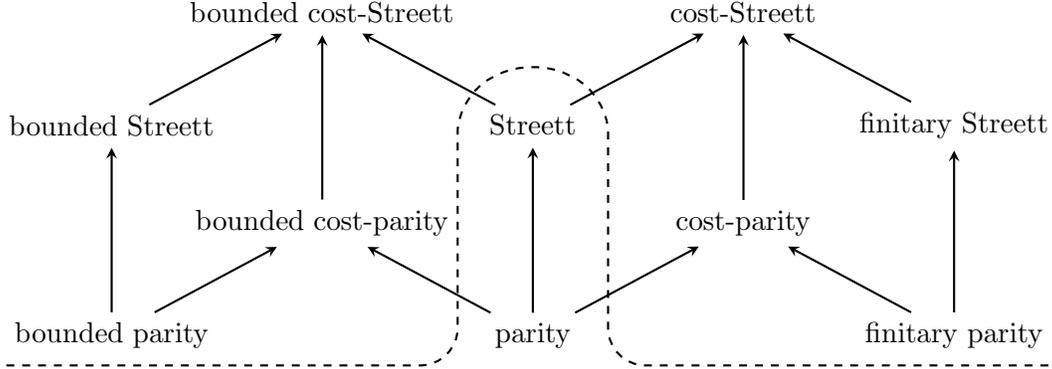

Finally, we obtain determinacy via the Borel determinacy 
theorem.

\begin{lem}
\label{remark_coststreettdeterminacy}
(Bounded) Streett games with costs are determined.
\end{lem}

\begin{proof}
Analogously to the proof of Lemma~\ref{lemma_costparitydeterminacy}.
\end{proof}

\section{Bounded Parity Games with Costs}
\label{section_solvingbcp}
In this section, we study bounded parity games with costs. We first show how
to solve such games, and then consider the memory requirements for winning
strategies for both players. Note that bounded parity games with costs
are a generalization of parity games, hence the algorithm we present
in the following subsection also has to solve parity games as a special case.

\subsection{Solving Bounded Parity Games with Costs via $\omega$-regular 
Games}

To solve bounded parity games with costs, we present a relaxation of the
bounded cost-parity condition, called $\pcrr$ which is a boolean combination of a parity, a co-B\"{u}chi, and a request-response~\cite{WallmeierHuettenThomas03} condition (hence its name). This condition essentially replaces the
bound~$b$ on the cost between a request and its response by just requiring an
answer to every request. Furthermore, for plays with finite cost it just requires the parity
condition to be satisfied, just as the bounded cost-parity condition does. The
$\pcrr$-condition is $\omega$-regular, thus both players have finite-state
winning strategies in games with $\pcrr$ winning conditions~\cite{BuechiLandweber69}. Using the fact that a finite-state winning strategy for
Player~$0$ answers every request within a fixed number of steps (and thereby
also with bounded cost), we are able to show that these two games have the
same winning regions. Finally, we show how to reduce the $\pcrr$-condition to
a parity condition. This completes our algorithm for solving bounded parity
games with costs and also yields upper bounds on the memory requirements of 
both players in bounded parity games with costs.

Let $\game = (\arena, \bcp(\col, \cost))$. First, we turn the cost function into a state property in order to be able to define cost-based winning conditions (which are sequences of vertices not edges): in the following, we assume that no
vertex of $\arena$ has both incoming increment- and $\epsilon$-edges. This can
be achieved by subdividing every increment-edge $e=(v,v')$: we add a new
vertex $\sub(e)$ and replace $e$ by $(v, \sub(e))$ (which is an
increment-edge) and by $(\sub(e), v')$ (which is an $\epsilon$-edge). Now,
only the newly added vertices have incoming increment-edges, but they do not
have incoming $\epsilon$-edges. Furthermore, it is easy to see that Player~$i$
wins from a vertex in the original game if and only if she wins from this 
vertex in the
modified game (where we color $\sub(e)$ by $\col(v')$). Finally, the
modification does not change the memory requirements, e.g., if Player~$0$
has a positional winning strategy for the modified game, then also for the
original game.

We say that a vertex is an increment-vertex, if it has an incoming
increment-edge (which implies that all incoming edges are increment-edges). 
Let
$I$ be the set of increment-vertices. Then, $\coB(I) = \set{\rho \mid
\cost(\rho) < \infty}$ is the set of infinite plays having finite cost,
i.e., those plays that visit only finitely many increment-vertices.
Furthermore, by $\rr(\col)$ we denote the set of infinite plays in which every
request is answered. We define
\[\pcrr(\col, I) = \left( \parity(\col) \cap \coB(I) \right) \cup \rr(\col)
\enspace,\]
which is $\omega$-regular, since it is a boolean combination of $\omega$-regular languages. Note that $\pcrr(\col, I)$ relaxes
$\bcp(\col,\cost)$ by giving up the bound on the cost between requests and
responses, in other words $\pcrr(\col, I) \supseteq \bcp(\col, \cost)$.

\begin{lem}
\label{lem_bcp2pcrr}
Let $\game = (\arena, \bcp(\col, \cost))$ and $\game' = (\arena, \pcrr(\col,
I))$, where $I$ is defined as above. A finite-state winning strategy for
Player~$i$ in $\game'$ from a set~$W$ of vertices is also a winning strategy for
Player~$i$ in $\game$ from $W$.
\end{lem}

\begin{proof}
The statement for $i = 1$ follows from the inclusion $V^\omega \setminus
\pcrr(\col, I) \subseteq V^\omega \setminus \bcp(\col, \cost)$.

Now, consider the case $i = 0$ and let $\sigma$ be a finite-state winning
strategy for Player~$0$ in $\game'$ from $W$. We argue that $\sigma$ is also a
winning strategy for Player~$0$ for $\game$ from $W$: let $\rho$ be consistent
with $\sigma$ and starting in $W$, which implies $\rho \in \pcrr(\col, I)$.

If $\rho$ satisfies $\parity(\col)$ and has only finitely many increments (say
$b$ many), then all but finitely many requests are answered with cost less
than $b + 1$ and there is no unanswered request of cost~$\infty$,
i.e., $\rho \in \bcp(\col, \cost)$.

Otherwise, $\rho$ satisfies $\rr(\col)$, i.e., every request in $\rho$ is
answered. We show that every request in $\rho$ is answered with cost at most
$b = |V| \cdot |\sigma|$ (where $|\sigma|$ is the size of the memory structure
implementing~$\sigma$), which implies that $\rho \in \bcp(\col, \cost)$.
Towards a contradiction, assume that there is a request that is answered with
cost greater than $b$. Then, there are two positions between the request and
its answer having the same vertex, an increment-edge in between them, and such
that the memory structure implementing~$\sigma$ assumes the same state at both
positions. Hence, using this loop forever is also a play that is consistent
with $\sigma$. However, this play contains an unanswered request of
cost~$\infty$ and therefore does not satisfy $\pcrr(\col, I)$. This yields the
desired contradiction to the fact that $\sigma$ is a winning strategy.
\end{proof}

\begin{cor}
\label{cor_bcp2pcrr_wr}
Let $\game$ and $\game'$ as in Lemma~\ref{lem_bcp2pcrr}. Then, $W_i(\game) =
W_i(\game')$ for $i \in \set{0,1}$.
\end{cor}

We now show how to reduce $\game' = (\arena, \pcrr(\col,I))$ to a 
parity game only linearly larger than $\game$. 
Let $O$ be the set of odd colors in $\col(V)$. We define the memory
structure~$\mem = (M, \init, \update)$ with $M = O \cup \set{ \bot}$,
\[\init(v) = \begin{cases}
\col(v) &\text{if $\col(v)$ odd,}\\
\bot    &\text{otherwise,}
\end{cases}\]
$\update(\bot, v) = \init(v)$, and 
\[\update(c, v) = \begin{cases}
\max(\col(v), c) &\text{if $\col(v)$ odd,}\\
\bot   			&\text{if $\col(v) \in \answer{c}$, } \\
c 				&\text{otherwise.}
\end{cases}\]
Intuitively, $\update^+(w)$ is the largest unanswered request in $w$, and is
$\bot$ if every request in $w$ is answered. Furthermore, let $\ell$ be an odd
color that is larger than every color in $\col(V)$. Now, we define a
coloring~$\col_\mem$ of the arena~$\arena \times \mem$ via
\[\col_\mem(v, m) = \begin{cases}
\ell + 1  & \text{if $m =\bot$,}\\
\ell   & \text{if $m \not=\bot$ and $v \in I$},\\ 
\col(v)   & \text{otherwise.}
\end{cases}\]
So, having all requests answered (i.e., being in memory state~$\bot$) is most
desirable for Player~$0$ while visiting increment-vertices (i.e., vertices in
$I$) while having an open request is most desirable for Player~$1$. If neither
of these occurs infinitely often, then the old coloring~$\col$ determines the
winner (without taking the memory states into account).

\begin{lem}
\label{lem_pcrr2parity}
Let $\game ' = (\arena, \pcrr(\col, I))$ and $\game'' = (\arena \times \mem,
\parity(\col_\mem))$. Then, $\game' \le_\mem \game''$.
\end{lem}

\begin{proof}
Let $\rho' = v_0 v_1 v_2 \cdots$ be a play in $\arena$ and $\rho'' = (v_0,
m_0) (v_1, m_1) (v_2, m_2) \cdots$ be its extended play in $\arena \times
\mem$. By construction, $m_j$ is the largest unanswered request in $v_0 \cdots
v_j$. We have to show that the same player wins both $\rho'$ in $\game'$ and
$\rho''$ in $\game''$.

Assume $\rho' \in \pcrr(\col, I)$. If $\rho' \in \rr(\col)$, then every
request is answered, i.e., $m_j$ is infinitely often equal to $\bot$. These
vertices have the largest color in $\game''$, which is even. Hence, $\rho''
\in \parity(\col_\mem)$. On the other hand, if $\rho' \in \parity(\col) \cap
\coB(I)$ but $\rho' \notin \rr(\col)$, then $\rho'$ and $\rho''$ each have a
suffix (starting after the last occurrence of an increment-vertex or the last unanswered request,
whichever comes last) such that
these suffixes have the same sequence of colors. Hence, $\rho''$ satisfies
$\parity(\col_\mem)$.

Conversely, assume $\rho'' \in \parity(\col_\mem)$. If $\ell+1$ is the maximal
color seen infinitely often, then $m_j$ is infinitely often equal to $\bot$,
which implies that every request in $\rho'$ is answered, i.e., $\rho' \in
\rr(\col) \subseteq \pcrr(\col, I)$. On the other hand, if the maximal color
seen infinitely often is smaller than $\ell+1$ (but still even, since we
assume Player~$0$ wins $\rho''$), then there are only finitely many
increment-vertices in $\rho'$ and the plays~$\rho'$ and $\rho''$ each have a
suffix such that these suffixes have the same sequence of colors. Hence,
$\rho'$ satisfies $\parity(\col)$. Altogether, we have $\rho' \in 
\parity(\col) \cap \coB(I) \subseteq \pcrr(\col, I)$. \end{proof}

\begin{cor}
\label{cor_memreq_bcp}
In bounded parity games with costs, both players have uniform finite-state
winning strategies of size $d+1$, where $d$ is the number of odd colors in the
game. 
\end{cor}

\begin{proof}
The reduction from games with winning condition $\pcrr(\col, I)$ to parity
games yields uniform finite-state winning strategies of size $d+1$ for such
games. Now apply Lemma~\ref{lem_bcp2pcrr}.
\end{proof}

In the next subsection, we show this bound to be tight for Player~$1$ and show
that Player~$0$ even has positional winning strategies.

The reduction from $\pcrr$ games to parity games and Lemma~\ref{lem_bcp2pcrr}
show that solving a parity game suffices to solve a bounded parity game with
costs and proves the following theorem. Here, $n$ is the number of vertices,
$m$ is the number of edges, and $d$ is the number of colors in the game.

\begin{thm}
\label{theorem_solvingbcp}
Given an algorithm that solves parity games in time $T(n,m,d)$, there is an
algorithm that solves bounded parity games with costs in time~$O(T(dn,
dm,d+2))$.
\end{thm}

Furthermore, since solving parity games is in $\np \cap \conp$ and the blowup
in our reduction is polynomial, we obtain the following remark (note that this was recently improved to $\up \cap \coup$~\cite{DBLP:conf/lpar/MogaveroMS13}).

\begin{rem}
\label{remark_compcomplexity_bcp}
The following problem is in $\np \cap \conp$: given a bounded parity game with 
costs~$\game$, $i\in\set{0,1}$, and a vertex $v$, is $v \in W_i(\game)$?
\end{rem}

Let us conclude by considering the special case of a bounded parity game with
costs~$\game$ in which every edge is an increment-edge, i.e., where $\game$ is
a bounded parity game. 
These games, called ``bounded parity games'' in~\cite{ChatterjeeHenzingerHorn09} 
can be solved in polynomial time. 
In this case, $\pcrr(\col, \cost)$ is
equal to $\rr(\col)$, which is a request-response
condition~\cite{WallmeierHuettenThomas03} where the sets of requests and
responses form a hierarchy, induced by the order on the colors. It is easy to
derive from the reduction to B\"{u}chi games~\cite{WallmeierHuettenThomas03}
that such games can be solved in polynomial time. Hence, we have recovered the
result of~\cite{ChatterjeeHenzingerHorn09} on bounded parity games as a
special case of our algorithm, although the running time of this algorithm is
worse than the running time of the algorithm presented
in~\cite{ChatterjeeHenzingerHorn09}.

\subsection{Memory Requirements in Bounded Parity Games with Costs}

In this subsection, we determine the exact memory requirements for both
players in bounded parity games with costs. We begin by considering Player~$0$
and improve on Corollary~\ref{cor_memreq_bcp}.

\begin{lem}
\label{lem_memory_bcp_pl0}
In bounded parity games with costs, Player~$0$ has uniform positional winning
strategies.
\end{lem} 

\begin{proof}
Due to Lemma~\ref{lem_bcp2pcrr}, it suffices to prove the statement for
games~$\game' = (\arena, \pcrr(\col, I))$. Recall that we reduced such a game
to a parity game~$\game'' = (\arena \times \mem, \parity(\col_\mem))$ using a
memory structure~$\mem$ that keeps track of the largest open request.
Specifically, Lemma~\ref{lem_pcrr2parity} reads as follows: $v_0 \in
W_0(\game')$ if and only if $(v_0,\init(v_0)) \in W_0(\game'')$.

We order $M = O \cup \set{\bot}$ with the natural order on integers for $O$,
where $\bot$ is the minimal element. Player~$0$'s winning region in $\game''$
is downwards-closed, i.e., $(v,m) \in W_0(\game'')$ and $m' < m$ implies
$(v,m') \in W_0(\game'')$, which can be shown by mimicking a winning strategy
from $(v,m)$ to also win from $(v,m')$. Thus, for $v \in W_0(\game')$, we
define 
\[\max(v) = \max\set{m \in M \mid (v,m) \in W_0(\game'')}\enspace,\] 
which is well-defined as $(v,\init(v)) \in W_0(\game'')$.

Now, let $\sigma''$ be a uniform positional winning strategy for Player~$0$ in
the parity game~$\game''$. We define a positional strategy $\sigma'$ for
$\game'$ by using $\max(v)$, i.e., the worst memory state Player~$0$ could be
in at vertex~$v$ while still being able to win from there. Given a vertex~$v
\in W_0(\game')$, let $\sigma''(v, \max(v)) = (v', m')$. Using this, we define 
$\sigma'(v) = v'$. We show that $\sigma'$ is a uniform winning strategy for
Player~$0$ in $\game'$. Consider a play $\rho' = v_0 v_1 v_2 \cdots$ starting
in $v_0 \in W_0(\game')$ consistent with $\sigma'$, and $\rho'' = (v_0,
m_0)(v_1, m_1)(v_2, m_2) \cdots$ its extended play in $\arena \times \mem$. A
straightforward induction shows that for every $j$, we have $(v_j, m_j) \in
W_0(\game'')$, so $\max(v_j) \ge m_j$. We have to show $\rho' \in \pcrr(\col,I)$.

By Lemma~\ref{lem_pcrr2parity}, $\rho' \in \pcrr(\col,I)$ if and only if
$\rho'' \in \parity(\col_\mem)$. Assume towards a contradiction that the maximal
color seen infinitely often in $\rho''$ is odd. This implies that the memory
state $\bot$ appears finitely often, so after a position, say $n$, all memory
states are different from $\bot$. Furthermore, from position $n$, we
additionally have $\max(v_j) \le \max(v_{j+1})$; this follows from the
observation that if $\update(c,v) \neq \bot$, then $c \le \update(c,v)$.
Consider $\rho^* = (v_n, \max(v_n)) (v_{n+1}, \max(v_{n+1}))(v_{n+2},
\max(v_{n+2})) \cdots$. Since the sequence $(\max(v_j))_{j \ge n}$ is
non-decreasing, it is ultimately constant, say from position $n' \ge n$. 
The suffix starting from $n'$ is consistent with $\sigma''$ and starts in $W_0(\game'')$, so it satisfies
$\parity(\col_\mem)$ since $\sigma''$ is a winning strategy. Consequently,
$\rho^*$ contains finitely many increment-vertices, so $\rho''$ as well. After
the last increment-vertex, $\rho^*$ and $\rho''$ have the same colors, but
$\rho''$ does not satisfy $\parity(\col_\mem)$, a contradiction. 
\end{proof}

To conclude this subsection, we prove that the upper bound~$d+1$ on the memory
requirements of Player~$1$ proved in Corollary~\ref{cor_memreq_bcp} is tight.

\begin{lem}
\label{lem_memory_bcp_pl1}
For every $d \ge 1$, there is a bounded parity game with costs~$\game_d$ such 
that
\begin{itemize}

\item the arena of $\game_d$ is of linear size in $d$ and there are $d$ odd
colors in $\game_d$,

\item Player~$1$ has a uniform finite-state winning strategy for $\game_d$ 
from every vertex, which is implemented with $d+1$ memory states, but

\item there is a vertex from which Player~$1$ has no winning strategy that is 
implemented with less than $d+1$ memory states.

\end{itemize}		
\end{lem}

\begin{proof}
We begin by describing the game by an example: Figure~\ref{fig_lowerboundsBP1}
depicts the game~$\game_4$, where the numbers in the vertices denote their
colors. Since each edge is an increment-edge, we do not label them as such in
the picture. The arena consists of a hub vertex colored by $0$ and four
disjoint blades, which are identified by the odd color of their outermost
vertex, i.e., by the colors $1,3,5$ and $7$ (which is $2 \cdot 4 - 1$). From
the hub, Player~$0$ can enter the blade for an odd color $c$ at a vertex of color~$c-1$
(which is even) which has a self-loop and an edge to a vertex of color~$8$
(which answers every request in the game). This vertex has only one outgoing
edge to a vertex of color~$c$ (this is the identifying color). Again, this
vertex has only one successor, the hub. In general, the arena of $\game_d$ has
$d$ blades, one for each color in $\set{1, 3, \ldots, 2d-1}$, the hub has
color $0$, and the second vertex in each blade has color $2d$ and thereby
answers every request. Furthermore, every edge is an increment-edge.

\begin{figure}
\begin{center}
\begin{tikzpicture}[scale = .9]
\begin{scope}
\tikzset{p0/.style = {shape = circle,    draw, thick, minimum size = 0.4cm}}
\tikzset{p1/.style = {shape = rectangle, minimum size =.4cm, draw, thick}}

\node[p0] 	at (0,0) 	(center)	{$0$};

\node[p1]	at (1.0,1.0)(11)		{$0$};
\node[p0]	at (2,2)	(12)		{$8$};
\node[p0]	at (3,3)	(13)		{$1$};

\node[p1]	at (1,-1)	(21)		{$2$};
\node[p0]	at (2,-2)	(22)		{$8$};
\node[p0]	at (3,-3)	(23)		{$3$};

\node[p1]	at (-1,-1)	(31)		{$4$};
\node[p0]	at (-2,-2)	(32)		{$8$};
\node[p0]	at (-3,-3)	(33)		{$5$};

\node[p1]	at (-1,1)	(41)		{$6$};
\node[p0]	at (-2,2)	(42)		{$8$};
\node[p0]	at (-3,3)	(43)		{$7$};

\end{scope}

\path[thick]
(center)	edge				(11)
(11)		edge				(12)
(12)		edge				(13)
(13)		edge[bend left]		(center)
(11)		edge[loop above]	()

(center)	edge				(21)
(21)		edge				(22)
(22)		edge				(23)
(23)		edge[bend left]		(center)
(21)		edge[loop right]	()

(center)	edge				(31)
(31)		edge				(32)
(32)		edge				(33)
(33)		edge[bend left]		(center)
(31)		edge[loop below]	()

(center)	edge				(41)
(41)		edge				(42)
(42)		edge				(43)
(43)		edge[bend left]		(center)
(41)		edge[loop left]	();

\end{tikzpicture}
\end{center}
\caption{The bounded parity game with costs~$\game_4$ (every edge is an increment-edge).}
\label{fig_lowerboundsBP1}
\end{figure}
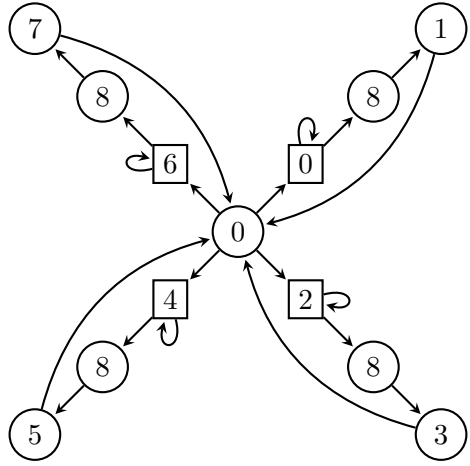

At the hub, Player~$0$ picks a blade (say of color~$c$) and then Player~$1$
decides whether to use the self-loop or to return to the hub. Note that
Player~$0$ loses, if she enters the blade of color~$c$ while there is an open
request of some color~$c' > c$, since Player~$1$ can use the self-loop of the
blade and thereby prevent an answer to the request~$c'$. On the other hand, if
Player~$1$ decides to leave the blade, all requests are answered and then
color~$c$ is requested. Note that this request is never answered to by moving
to the hub.

First, we show that Player~$1$ has a uniform finite-state winning strategy
from every vertex that is implementable with $d +1$ memory states. The memory
structure keeps track of the largest open request, i.e., we use states~$1, 3,
\ldots, 2d -1$ and an additional state~$\bot$ that is reached, if there is no
open request. Now, assume the current memory state is $m$ and the play is in a
vertex of Player~$1$, which is uniquely identified by its even color~$c$. If
$m = \bot$, then Player~$1$ moves to the (unique) successor of color~$2d$.
Now, assume $m \not= \bot$, i.e., $m$ is some odd color. If $m < c$, then
Player~$1$ again leaves $c$ by moving to the unique successor of color~$2d$.
If $m > c$, then Player~$1$ uses the self-loop forever.

Now, consider a play that is consistent with this strategy. If the current
memory state is $\bot$, then a request is raised within the next three moves
and the play returns to the hub, which implies that the memory is updated to
some odd color~$m$. From there, Player~$0$ has to move to some blade, say for
color~$c$ (which is odd). If $m > c$, then Player~$1$ uses the self-loop at
the vertex of color~$c - 1$ forever. The resulting play is winning for him,
since the request of $c$ is unanswered with cost~$\infty$. On the other hand,
if $m < c$, then Player~$1$ moves to the vertex of color~$c$ and then back to
the hub. While doing this, the memory is updated to a larger state, namely
$c$. Hence, the memory states along a play consistent with the strategy
described above are increasing, which means that at some point Player~$0$ has
to enter a blade for color $c < m$, where $m$ is the current memory state,
i.e., also an open request. Then, Player~$1$ will win by using the self-loop
of this blade. Hence, the strategy described above is a winning strategy from
every vertex and is implemented using $d+1$ memory states.

It remains to show that the upper bound~$d+1$ is tight. To this end, consider
a finite-state strategy~$\tau$ for Player~$1$ that is winning from the hub,
say $\tau$ is implemented by $(M, \init, \update)$. We show that $M$ contains
at least $d+1$ memory states. To this end, we define a sequence $m_0, m_1,
\ldots, m_d$ of $d+1$ memory states, as follows. Define $m_0 = \init(v)$,
where $v$ is the hub. Now, Consider the play where Player~$0$ moves from the 
hub to
the blade with color~$1$. Since $\tau$ is a winning strategy, Player~$1$ will
use the self-loop of this blade only finitely often, i.e., the hub is reached
again. We denote this play prefix by $w_1$ (which is consistent with $\tau$)
and define $m_1 = \update^+(w_1)$. Consider now the play where after $w_1$,
Player~$0$ moves to the blade with color~$3$. Again, Player~$1$ will use the
self-loop only finitely often and the hub is reached again. We denote the
prolongation of $w_1$ through this blade by $w_2$ and define $m_2 =
\update^+(w_1w_2)$. This process is continued for each blade in ascending
order. Since Player~$1$ has to leave each blade we obtain a sequence $m_0,
m_1, \ldots, m_d$ of memory states assumed at the visits of the hub and a play
prefix~$w_1w_2 \cdots w_{d}$ that is consistent with $\tau$, starts in the
hub, and satisfies $\update^+(w_1 w_2 \cdots w_j) = m_j$ for every $j\ge 1$.
Furthermore, each $w_1 w_2 \cdots w_j$ ends in the hub.

We argue that the states $m_0, m_1, \ldots, m_d$ are pairwise distinct. Assume
towards contradiction there are $j < j' \le d$ such that $m_j = m_{j'}$. Then,
the play $\rho = w_1 \cdots w_{j} \cdot ( w_{j+1} \cdots w_{j'} )^\omega$ is
consistent with $\tau$. However, the maximal color seen infinitely often
during $\rho$ is $2d$ (which answers every request), and there is a uniform
bound on the distance between the occurrences of $2d$. Hence, the play is
winning for Player~$0$ in $\game_d$, contradicting the fact that $\tau$ is a
winning strategy for Player~$1$. Hence, the states $m_j$ are indeed pairwise
distinct. Thus, every winning strategy has at least $d+1$ memory states.
\end{proof}

Recall that every edge in $\game_d$ is an increment-edge, i.e., $\game_d$ is a
bounded parity game. In~\cite{ChatterjeeHenzingerHorn09} an upper bound of two
on the memory requirements of Player~$1$ is claimed for bounded parity games.
The games presented here refute this claim: there is no constant bound on the
memory needed for Player~$1$ in bounded parity games.

\section{Solving Parity Games with Costs via Bounded Parity Games with Costs}
\label{section_solvingcp}

In this section, we show that being able to solve bounded parity games with
costs suffices to solve parity games with costs. Our algorithm is based on the
following lemma which formalizes this claim by relating the winning regions of
Player~$0$ in a parity game with costs and the bounded parity game with costs in the same arena. The algorithm presented here is equal to the one presented to solve finitary parity games by solving bounded parity games~\cite{ChatterjeeHenzingerHorn09}. However, our correctness proof is more general, since it has to deal with plays of finite cost, which do not exist in finitary and bounded parity games.

We begin by relating the winning regions of a parity game with costs and the bounded parity game with costs in the same arena.

\begin{lem}
\label{lem_cp2bcp_fixpoint}
Let $\game = (\arena, \cp(\col, \cost))$ and let $\game' = (\arena, \bcp(\col,
\cost))$.
\begin{enumerate}

\item\label{lem_cp2bcp_fixpoint_playerzerosubset} 
$W_0( \game' ) \subseteq W_0(\game)$.

\item\label{lem_cp2bcp_fixpoint_playerzeroempty} 
If $W_0( \game' ) = \emptyset$, then  $W_0( \game) = \emptyset$.

\end{enumerate}
\end{lem}

\begin{proof} 
(\ref{lem_cp2bcp_fixpoint_playerzerosubset}) 
This follows from the inclusion $\bcp(\col,\cost) \subseteq 
\cp(\col,\cost)$.

(\ref{lem_cp2bcp_fixpoint_playerzeroempty})
Due to determinacy, if $W_0(\game') = \emptyset$, then we have $W_1(\game') =
V$. Due to Corollary~\ref{cor_memreq_bcp}, Player~$1$ has a uniform
finite-state strategy~$\tau'$ that is winning from every vertex~$v$ in $\game'$. Consider a
play consistent with $\tau'$: either, for every $b\in\nats$, there is a request that is
open for the next $b$ increment-edges (the request could be the same for every
$b$), or the maximal color seen infinitely often is odd (i.e., there are
infinitely many unanswered requests). 

To win in the cost-parity game, Player~$1$ has to keep for every $b$ a different request open for at least $b$ increment-edges (or violate the parity condition). We define a strategy~$\tau$ for Player~$1$ in $\game$ that achieves this by restarting $\tau'$ every time a bound $b$ is exceeded. To this end, $\tau$ is guided by a counter~$b$ which is initialized with 1 and the strategy $\tau$ behaves like $\tau'$ until a request is open for $b$ increment-edges. If this is the case, $b$ is incremented and $\tau$ behaves like $\tau'$ does when it starts from the current vertex (forgetting the history of the play constructed so far).

Formally, $\tau$ is implemented by the infinite memory structure~$\mem = (M, \init, \update)$ where $M = \nats \times V^+$, $\init(v) = (1,v)$, and $\update((b,w),v)$ is defined as follows: if $w$ contains a request that is open for more than $b$ increment-edges, then $\update((b,w),v) = (b+1, v)$, i.e., the counter is incremented and the play prefix in the second component is reset. On the other hand, if $w$ does not contain a request that is open for more than $b$ increment-edges, then $\update((b,w),v) = (b, wv)$, i.e., the counter is unchanged and the vertex~$v$ is appended to the play prefix. Note that we always have $\update^*(\rho_0 \cdots \rho_k) = (b, \rho_{k'} \cdots \rho_k)$ for some non-empty suffix~$\rho_{k'} \cdots \rho_k$ of $\rho_{0} \cdots \rho_k$. Hence, we can define $\tau(\rho_0 \cdots \rho_k) = \tau'(\rho_{k'} \cdots \rho_k)$.

We show that $\tau$ is winning in $\game$ from every vertex, which implies
$W_0(\game) = \emptyset$. Let $\rho$ be a play that is consistent with
$\tau$ and distinguish two cases: if the counter in the first component of the memory states reached during $\rho$ is incremented infinitely
often, then $\rho$ contains for every~$b\in\nats$ a request that is open for at least
$b$ increment-edges, so $\rho \notin \cp(\col,\cost)$. On the other hand, if
the counter is incremented only finitely often (say to value $b$), then
there is a suffix~$\rho'$ of $\rho$ that is consistent with the strategy
$\tau'$. Since the counter is not incremented during $\rho'$, every request in
$\rho'$ is either answered with cost at most~$b$ or not answered, but only
followed by at most $b$ increment-edges. Hence $\tau'$ ensures that the maximal
color seen infinitely often in $\rho'$ is odd, i.e., $\rho' \notin
\parity(\col)$, so $\rho \notin \parity(\col)$, and a fortiori $\rho \notin
\cp(\col,\cost)$. 
\end{proof}

We have seen in
Example~\ref{example_costparity} that Player~$1$ needs infinite memory to win cost-parity games. Indeed, the winning strategy for Player~$1$ described in the example
proceeds as the strategy~$\tau$ described above. It requests color~$1$ at vertex~$a$, uses the
loop at vertex~$b$ to keep the request unanswered for several steps and then
forgets about this request. At this point, a new request has to be raised by
moving from $b$ back to $a$, thereby answering the old request at vertex~$c$.
This request is then kept unanswered for more increment-edges than the previous
one, and this goes on ad infinitum.

To conclude this subsection, we show how Lemma~\ref{lem_cp2bcp_fixpoint} can
be used to solve parity games with costs. Let $\game = (\arena,\cp(\col,
\cost))$. The following algorithm proceeds by iteratively removing parts of
$\arena$ that are included in the winning region of Player~$0$ in $\game$: we
have just proven that the winning region of Player~$0$ in the bounded parity
game with costs in $\arena$ is a subset of her winning region in the parity
game with costs in the same arena. Thus we can remove it and its attractor.
This is repeated until Player~$0$'s winning region in the bounded parity game
with costs is empty. In this case, her winning region in the parity game with
costs is empty as well, again due to Lemma~\ref{lem_cp2bcp_fixpoint}. This
idea is implemented in the following algorithm.

\begin{algorithm}
\begin{algorithmic}
 
 \STATE {$j \leftarrow 0$}; $W_0^j \leftarrow \emptyset$; $\arena_j \leftarrow \arena$ 
 \REPEAT
 \STATE {$j \leftarrow j + 1$}
 \STATE {$X_j \leftarrow W_0(\arena_{j-1},\bcp(\col, \cost))$}
 \STATE {$W_0^j \leftarrow W_0^{j-1} \cup \att{0}{\arena_{j-1}}{X_j}$}
 \STATE {$\arena_j \leftarrow \arena_{j-1} \setminus
 \att{0}{\arena_{j-1}}{X_j}$}
 \UNTIL {$X_j = \emptyset$}
 \RETURN{$W_0^j$}
\end{algorithmic} 
\caption{A fixed-point algorithm for solving $(\arena, \cp(\col, \cost))$.}
\label{algorithm_fixpoint}
\end{algorithm}

\begin{exa} 
Running on the parity game with costs of
Figure~\ref{figure_arena}, Algorithm~\ref{algorithm_fixpoint} computes
$X_1=\set{g}$ and $W_0^1 = \set{f,g}$, $X_2 = \set{e}$ and $W_0^2 =
\set{d,e,f,g}$, and $X_3 = \emptyset$. Thus, it returns $W_0^2$, which is the
winning region of Player~$0$ in the game. 
\end{exa}

Next, we show the algorithm to be correct and bound its number of iterations. 

\begin{lem}
\label{lem_fixpointalgo_correctness}
Let $\game$ be a parity game with costs with $n$ vertices. 
Algorithm~\ref{algorithm_fixpoint} returns the winning region~$W_0(\game)$ 
after at most $n+1$ iterations.
\end{lem}

\begin{proof}
Let $\game = ((V, V_0, V_1, E), \cp(\col, \cost))$ and let $t$ be the last
iteration of Algorithm~\ref{algorithm_fixpoint} with $X_t \not= \emptyset$, i.e., the algorithm returns $W_0^t$. We
have $t \le |V|$, since $\emptyset = W_0^0 \subsetneq W_0^1 \subsetneq \cdots
\subsetneq W_0^t \subseteq V$ is a strictly increasing chain. Hence, the
algorithm terminates after at most $|V|+1$ iterations.

Next, we show $W_0^t \subseteq W_0(\game)$: to this end, we define a
strategy~$\sigma$ for Player~$0$ on $W_0^t$ as follows: on the sets~$X_j$ computed by the algorithm, which are
winning regions of Player~$0$ in a bounded parity game with costs, we play
using some uniform positional winning strategy for this game, which always
exists due to Lemma~\ref{lem_memory_bcp_pl0}. On the
attractors~$\att{0}{\arena_{j-1}}{X_j}$ we play using some positional
attractor strategy. Thus, $\sigma$ is a positional strategy that is defined
for every vertex in $W_0^t$. Next, we show that it is indeed a uniform
positional winning strategy from $W_0^t$.

Every winning region~$X_j$ is a trap for Player~$1$ in $\arena_{j-1}$. Hence,
in the whole arena~$\arena$, Player~$1$ can leave $X_j$ only to vertices in
some $W_0^{j'}$ with $j'<j$. Player~$0$ on the other hand only moves to vertices
in $W_0^{j}$. Similarly, if the play is in the attractor of some $X_j$, then it
reaches $X_j$ after at most $|V|$ steps or Player~$1$ moves to some $W_0^{j'}$
for some $j' <j$. Hence, every play~$\rho$ consistent with $\sigma$ has a
suffix~$\rho'$ that visits only vertices from some $X_j$ and is consistent
with the winning strategy for the corresponding bounded parity game with
costs. So, $\rho' \in \bcp(\col, \cost)$, which implies $\rho \in \cp(\col,
\cost)$. Thus, $\sigma$ is a winning strategy for Player~$0$ in $\game$ from $W_0^t$.

It remains to consider Player~$1$, i.e., to show that $V \setminus W_0^t
\subseteq W_1(\game)$. Note that $V \setminus W_0^t$ is the set of vertices of
$\arena_t$ and that we have $W_0(\arena_t, \bcp(\col, \cost)) = \emptyset$.
Hence, by
Lemma~\ref{lem_cp2bcp_fixpoint}(\ref{lem_cp2bcp_fixpoint_playerzeroempty}) we
conclude $W_0(\arena_t, \cp(\col, \cost)) = \emptyset$, i.e., Player~$1$ wins
the parity game with costs in the arena~$\arena_t$ from every vertex. Since $V
\setminus W_0^t$ is a trap for Player~$0$ (it is the complement of an
attractor), it follows that Player~$1$ wins the parity game with costs in the
arena~$\arena$ from every vertex in $V \setminus W_0^t$. \end{proof}

\begin{cor}
\label{cor_memory_cp_pl1}
In parity games with costs, Player~$0$ has uniform positional winning
strategies.
\end{cor}

Using Lemma~\ref{lem_fixpointalgo_correctness},
Theorem~\ref{theorem_solvingbcp}, and the fact that
Algorithm~\ref{algorithm_fixpoint} terminates after at most $n+1$ iterations,
and therefore has to solve at most $n$ bounded parity games with costs, we
obtain the following result where again $n$ is the number of vertices, $m$ is
the number of edges, and $d$ is the number of colors in the game.

\begin{thm}
\label{theorem_solvingcp}
Given an algorithm that solves parity games in time $T(n,m,d)$, there is an 
algorithm that solves parity games with costs in time~$O(n \cdot T(dn, 
dm,d+2))$.
\end{thm}

Also, we obtain the same computational complexity as for bounded parity games
with costs. Here, we rely on the characterization of the winning region~$W_0(\game)$ of a parity game with costs~$\game$ as computed by
Algorithm~\ref{algorithm_fixpoint}: the sets $X_j$ can be determined in $\np$
(respectively in $\conp$) due to Remark~\ref{remark_compcomplexity_bcp} and
the attractors can be computed in (deterministic) linear time (note that this was recently improved to $\up \cap \coup$~\cite{DBLP:conf/lpar/MogaveroMS13}).

\begin{rem}
\label{remark_compcomplexity_cp}
The following problem is in $\np \cap \conp$: given a parity game with 
costs~$\game$, $i\in\set{0,1}$, and a vertex $v$, is $v \in W_i(\game)$?
\end{rem}

In the previous section we have shown that one can recover a
polynomial time algorithm for deciding the winning regions of bounded parity
games. Hence, using Algorithm~\ref{algorithm_fixpoint}, we obtain the same for
finitary parity games as well. Hence, we also recover polynomial time
decidability of finitary parity games as a special case of our algorithm.
However, this is not surprising since Algorithm~\ref{algorithm_fixpoint} is
the same one used in~\cite{ChatterjeeHenzingerHorn09} to solve finitary parity
games via solving bounded parity games.

\section{Positional Winning Strategies for Bounded Parity Games with Costs\\ 
via Scoring Functions}
\label{section_elimmem}

In Lemma~\ref{lem_memory_bcp_pl0}, we have shown how to eliminate the
memory introduced in the reduction from $\pcrr$ games to parity games,
which proved the existence of uniform positional winning strategies
for Player~$0$ in bounded parity games with costs. Using these
strategies as building blocks, we also proved the existence of uniform
positional winning strategies for Player~$0$ in parity games with
costs. Intuitively, the memory used in the reduction from bounded
parity with costs to $\pcrr$ keeps track of the largest open request,
but Player~$0$ does not need this information to implement her winning
strategy as proved in Lemma~\ref{lem_memory_bcp_pl0}. Instead, she can
always play assuming the worst situation that still allows her to
win. Thus, we have shown that the memory introduced by this reduction
can always be eliminated.

In this section we generalize this construction to memory
structures that are not necessarily of the form used in the reduction: we show how to turn an arbitrary uniform finite-state winning strategy for Player~$0$ in a bounded parity game with costs into a positional one. To this end, we define a quality measure for play prefixes and then show that
always playing like in the worst possible situation is a positional winning
strategy. This gives an alternative proof of half-positional determinacy\footnote{A game is half-positionally determined, if one of the players has a positional winning strategy from every vertex of her winning region.} of
(bounded) parity games with costs and presents a general framework that we
believe to be applicable to other winning conditions as well.

We begin by defining a so-called scoring function for bounded parity games
with costs that measures the quality of a play prefix (from Player~$0$'s
vantage point) by keeping track of the largest unanswered request, the number
of increment-edges traversed since it was raised, and how often each odd color
was seen since the last increment-edge. This information is gathered in a so-called score-sheet, which is then used to measure the quality of the play prefix. We begin by defining score sheets.

For the remainder of this section, we
fix a bounded parity game with costs $\game = (\arena, \bcp(\col, \cost))$
with arena~$\arena = (V, V_0, V_1, E)$, and an arbitrary uniform finite-state winning strategy~$\sigma$ for Player~$0$ in $\game$ which we want to turn into a uniform positional winning strategy. Let $\col(V) \subseteq \set{0, 1,\ldots, \ell}$, where we assume $\ell$ to be odd. Furthermore, let $d = \frac{\ell+1}{2}$ be the number of odd colors in $\set{0,1,\ldots,\ell}$.
Finally, let $t = |V|\cdot |\sigma|$, where $|\sigma|$ denotes the size of the
memory structure implementing $\sigma$.

A proper (score-) sheet is a vector $(c, n, s_\ell, s_{\ell-2}, \ldots, s_3,
s_1)$ where $c$ is an odd color in $\set{1, 3, \ldots, \ell}$, $n \le t$, and
$s_{c'} \le t$ for every $c'$. Finally, we use two non-proper sheets denoted by
$\bot$ and $\top$. The reversed ordering of the score values $s_\ell,
s_{\ell-2}, \ldots, s_3, s_1$ in the sheets is due to the max-parity
condition, in which larger colors are more important than smaller ones. This
is reflected by the fact that we compare sheets in the lexicographical order
induced by $<$ on its components and add $\bot$ as minimal and $\top$ as
maximal element. For example, $(3, 3, 0, 1, 1) \lexordstrict (3,3, 1, 0, 3)$
and $\bot \lexordstrict s \lexordstrict \top$ for every sheet~$s \not= \bot,
\top$. As usual, we write $s \lexord s'$ if $s = s'$ or $s \lexordstrict s'$.

Next, we show how to update sheets along a play to use them as a quality
measure for play prefixes. Let $s = (x_1, \ldots, x_{d+2})$ be a proper
sheet. We say that $s$ is full in coordinate~$1$, if $x_1 = \ell$ (recall that
$\ell$ is the largest possible value in the first coordinate), and that $s$ is
full in coordinate $k>1$, if $x_k = t$ (recall that $t$ is the largest
possible value in all but the first coordinate). Let $k$ be a coordinate and
let $s = (x_1,\ldots, x_{d+2})$ be a proper sheet.

\begin{itemize}

\item If $1$ is the largest coordinate smaller than or equal to $k$ that is
not full in $s$, then incrementing $s$ at coordinate~$k$ yields the
sheet~$(x_1+2, 0, \ldots, 0)$. If $k>1$, then we say that there is an overflow
in coordinates $2, \ldots, k$.

\item If $k' > 1$ is the largest coordinate smaller than or equal to $k$ that
is not full in $s$, then incrementing $s$ at coordinate~$k$ yields the
sheet~$(x_1, \ldots, x_{k'-1}, x_{k'}+1, 0, \ldots, 0)$. If $k' < k$, then we
say that there is an overflow in coordinates $k' +1 , \ldots, k$.

\item If there is no coordinate $k'$ smaller than or equal to $k$ that is not
full in $s$, then incrementing $s$ at coordinate~$k$ yields the sheet~$\top$
and we say that there is an overflow in coordinates $1, \ldots, k$.

\end{itemize}

\begin{exa}
Assume we have $\ell = 5$ and $t = 3$ and consider $s = (3, 3, 0, 1, 3)$.
Then, $s$ is full in coordinate~$2$ and $5$, but not in coordinates~$1$, $3$,
and $4$. Incrementing $s$ at coordinate~$1$ or $2$ yields the sheet~$(5, 0, 0,
0, 0)$ (note that there is an overflow of coordinate~$2$ in the second case),
incrementing at $3$ yields $(3, 3, 1, 0, 0)$ while incrementing at $4$ or $5$
yields $(3,3,0,2,0)$ (and there is an overflow of coordinate~$5$ in the second
case). 
\end{exa}

Next, we show that the increment-operation and a reset-operation are
compatible with the ordering. Recall that we compare sheets lexicographically.

\begin{rem}
\label{remark_sheetpropertiesabstract} 
Let $x = (x_1, \ldots, x_{d+2}) \lexord
y = (y_1, \ldots, y_{d+2})$ be two sheets and let $k$ be a coordinate.

\begin{enumerate}
\item	\label{remark_sheetpropertiesabstract_inccong}
Let $x'$ (respectively $y'$) be obtained by incrementing $x$ (respectively 
$y$) at coordinate $k$. Then, $x' \lexord y'$.

\item \label{remark_sheetpropertiesabstract_resetcong}
Let $x'' = (x_1, \ldots, x_k, 0, \ldots, 0)$ and $y'' = (y_1, \ldots, y_k, 0, 
\ldots, 0)$. Then, $x'' \lexord y''$.

\end{enumerate}
\end{rem}

Now, we want to assign a sheet to every play prefix. To this end, we define 
the initial sheet~$\sheet(v)$ of a vertex~$v$ by 
\begin{equation*}
\sheet(v)=
\begin{cases}
\bot                                    & \text{if $\col(v)$ is even,}\\
(\col(v), 0, 0, \ldots, 0)              & \text{if $\col(v)$ is odd.}
\end{cases}
\end{equation*}
Now, let $\sheet(wv)$ for $w \in V^*$ and $v \in V$ be already defined and let
$(v,v')$ be an edge. If $\sheet(wv) = \top$, then $\sheet(wvv') = \top$, and
if $\sheet(wv) = \bot$, then $\sheet(wvv') = \sheet(v')$. Now, assume we have
$\sheet(wv) = (c, n, s_\ell, \ldots, s_1)$, i.e., $c$ is the largest open request in $wv$. We have to distinguish several
cases. 
\begin{itemize}

\item If $\col(v') > c$, then $\sheet(wvv') = \sheet(v')$, i.e., if $\col(v')$ is even larger than $c$ and odd, then the first component is updated to $\col(v')$ and all others are reset to zero. If $\col(v')$ is even and larger than $c$, then all requests are answered and the sheet is reset to $\bot$.

\item If $\col(v') \le c$ and $\cost(v,v') = i$ (i.e., the largest open request is still $c$ but an increment-edge is traversed), then $\sheet(wvv')$ is
obtained from $\sheet(wv)$ by incrementing the second coordinate (the one
associated with costs).

\item if $\col(v') \le c$, $\cost(v,v') = \epsilon$ and $\col(v')$ even, then
the scores for the colors that are answered by $\col(v')$ are reset to zero,
i.e., $\sheet(wvv') = (c, n, s_\ell, \ldots, s_{\col(v')+1}, 0, \ldots, 0)$,

\item if $\col(v') \le c$, $\cost(v,v') = \epsilon$ and $\col(v')$ odd, then
$\sheet(wvv')$ is obtained from $\sheet(wv)$ by incrementing the coordinate
storing the score for $\col(v')$. 

\end{itemize}
Note that the increments in the second and fourth case of the definition might trigger overflows in case the respective coordinates are full in $\sheet(wv)$.

Let $\sheet(w) = (c, n, s_\ell, \ldots, s_1)$. To simplify our notation in the
following proofs, we define $\req(w) = c$, $\creq(w) = n$, and $\score_{c'}(w)
= s_{c'}$. If $\sheet(w) = \bot$ or $\sheet(w) = \top$, then we leave these
functions undefined. Furthermore, let $\last(w)$ denote the last vertex of a
non-empty finite play $w$. 

In the following, we show three properties of the scoring function that
are used to prove our main result. We begin by showing that it is a congruence.

\begin{lem}
\label{lemma_congurence}	

If $\last(x) = \last(y)$ and $\sheet(x) \lexord \sheet(y)$, then $\sheet(xv)
\lexord \sheet(yv)$ for every $v \in V$.
\end{lem}

Before we begin the proof we state the following useful facts.

\begin{rem}
\label{remark_sheetprops}
Let $w \in V^*$ and $v \in V$.
\begin{enumerate}

\item\label{remark_sheetprops_update}
If $\req(w) \not= \req(wv)$, then $\sheet(wv) = \sheet(v)$.

\item\label{remark_sheetprops_sheetgreaterlastvertex}
If $\col(v)$ is odd, then $\sheet(wv) \ge \sheet(v)$.

\end{enumerate}
\end{rem}

Now, we are ready to prove Lemma~\ref{lemma_congurence}.

\begin{proof}
If $\sheet(x) = \sheet(y)$, then $\sheet(xv) = \sheet(yv)$, since the sheets
of $xv$ and $yv$ only depend on the sheets of $x$ and $y$ (which are equal)
and the last edges of $xv$ and $yv$ (which are also equal).

So, consider the case $\sheet(x) \lexordstrict \sheet(y)$. First, assume we
have $\sheet(x) = \bot$, which implies $\sheet(xv) = \sheet(v)$. If $\col(v)$
is even, then $\sheet(v) = \bot$ and we are done, since $\bot$ is the minimal
element. Otherwise, applying
Remark~\ref{remark_sheetprops}(\ref{remark_sheetprops_sheetgreaterlastvertex})
to $yv$ yields the desired result. As a last special case assume we have
$\sheet(y) = \top$, which implies $\sheet(yv) = \top$. As $\top$ is the
maximal element, we have $\sheet(xv) \lexord \sheet(yv)$.

We are left with the case $\bot \lexordstrict \sheet(x) \lexordstrict
\sheet(y) \lexordstrict \top$ and have to consider two subcases:

\begin{enumerate}
\item First, assume we have $\req(x) = \req(y)$. We consider several subcases.
\begin{enumerate}

\item If $\col(v) > \req(x) = \req(y)$, then $\sheet(xv) = \sheet(yv) =
\sheet(v)$.

\item If $\col(v) \le \req(x) = \req(y)$ and $\cost(\last(x),v) = i$, then
both $\sheet(xv)$ and $\sheet(yv)$ are obtained by incrementing $\sheet(x)$
and $\sheet(y)$ respectively at the second coordinate. Hence,
Remark~\ref{remark_sheetpropertiesabstract}(\ref{remark_sheetpropertiesabstract_inccong})
yields the desired result.

\item If $\col(v) \le \req(x) = \req(y)$, $\cost(\last(x),v) = \epsilon$, and
$\col(v)$ is even, then both $\sheet(xv)$ and $\sheet(yv)$ are obtained by
reseting the scores for every $c'$ smaller than $\col(v)$. Hence,
Remark~\ref{remark_sheetpropertiesabstract}(\ref{remark_sheetpropertiesabstract_resetcong})
yields the desired result.

\item If $\col(v) \le \req(x) = \req(y)$, $\cost(\last(x),v) = \epsilon$, and
$\col(v)$ is odd, then both $\sheet(xv)$ and $\sheet(yv)$ are obtained by
incrementing $\sheet(x)$ and $\sheet(y)$ respectively at the coordinate
storing the score for $\col(v)$. Hence,
Remark~\ref{remark_sheetpropertiesabstract}(\ref{remark_sheetpropertiesabstract_inccong})
yields the desired result. \end{enumerate}

\item Now, assume we have $\req(x) < \req(y)$. Note that $\sheet(xv) = \top$
is impossible in this case, since the first coordinate of $x$ is not full, as
it is strictly smaller than $\req(y)$. We again have to consider several
subcases:

\begin{enumerate}
\item If $\col(v) > \req(y) > \req(x)$, then $\sheet(xv) = \sheet(yv) =
 \sheet(v)$.

\item If $\col(v) = \req(y) > \req(x)$, then we have $\req(xv) > \req(x)$ and
an application of
Remark~\ref{remark_sheetprops}(\ref{remark_sheetprops_update}) and
\ref{remark_sheetprops}(\ref{remark_sheetprops_sheetgreaterlastvertex}) (to
$yv$) yields the desired result.

 \item Finally, assume we have $\col(v) < \req(y)$. We again have to consider
three subcases.
\begin{enumerate}

\item If we have $\sheet(yv) = \top$, then we are done.

\item Assume we have $\req(yv) > \req(y)$. Then the following inequalities hold:
$\req(xv) \le \req(x) +2 \le \req(y) < \req(yv)$, where the first one is due
to the fact that the first component of $\sheet(xv)$ can only increase due
to an increment, and the second one due to $\req(x) < \req(y)$. Hence, we have
$\sheet(xv) \lexordstrict \sheet(yv)$ in this case.

\item Finally, consider the case where $\req(yv) = \req(y)$. If $\req(xv) <
\req(yv)$, then we are done. So, assume we have $\req(xv) = \req(yv)$. Then,
the first component of $\sheet(x)$ is increased to obtain $\sheet(xv)$, which
implies that all other components of $\sheet(xv)$ are equal to zero. Hence, we
have $\sheet(xv) \lexord \sheet(yv)$.\qedhere

\end{enumerate} 
\end{enumerate}
\end{enumerate}
\end{proof}\smallskip

\noindent We continue by showing that the sheets of a play $\rho$ being bounded is a
sufficient condition for $\rho$ satisfying the bounded cost-parity condition.

\begin{lem}
\label{lemma_boundedwins}
If the sheets of all prefixes of a play~$\rho$ are strictly smaller than
$\top$, then $\rho \in \bcp(\col, \cost)$. \end{lem}

\begin{proof}
We prove the converse, i.e., if $\rho \notin \bcp(\col, \cost)$, then there is
a prefix of $\rho$ whose sheet is $\top$. First, assume that for every $b$
there is a request (say of color~$c$) that is open for at least $b$
increment-edges. Then, the second component of the sheets is incremented every
time an increment-edge is traversed before the request is answered. Also, the
first component is increased every time the second component overflows or
every time a larger odd color is visited. Note that it is not reset in this
interval, as the request of $c$ is not answered. Hence, if we pick $b$ large
enough, the first component overflows as well. This yields the sheet $\top$.

Now assume the maximal color seen infinitely often, call it $c$, is odd. We
may assume that $\rho$ has only finitely many increment-edges, as we are in
the first case otherwise. Pick a position of $\rho$ such that the maximal
color appearing after this position is $c$ and such that no increment-edge is
traversed after this position. After this position, the coordinate storing the
score for $c$ is incremented again and again. Furthermore, this coordinate
(and all to the left of this one) are only reset in case of an overflow, which
means that there is a coordinate to the left that is incremented. Thus, every
coordinate to the left of the one storing the score for $c$ is incremented
again and again, too. Hence, the first component overflows at some point,
which yields the sheet $\top$. \end{proof}

Recall that the entries in all but the first component of a (proper) sheet are
bounded by $t = |V|\cdot|\sigma|$, where $\sigma$ is a uniform finite-state
winning strategy for Player~$0$ in $\game = (\arena, \bcp(\col, \cost))$.
Next, we show that this strategy keeps the sheets smaller than~$\top$.

\begin{lem}
\label{lemma_fsbounds}
Let $\rho$ be starting in $W_0(\game)$ and consistent with
$\sigma$. Then, the sheets of all prefixes of $\rho$ are strictly smaller than 
$\top$.
\end{lem}

\begin{proof}
First, we show that every request in $\rho$ is answered with cost less than or
equal to $t$ or followed by at most $t$ increment-edges. Towards a
contradiction, assume there is a request that is followed by $t +1$
increment-edges, but no answer before the last of these increment-edges. Then,
there are two positions in this interval that have the same vertex, the memory
structure implementing $\sigma$ assumes the same state after both positions,
and there is at least one increment-edge between these positions. Hence, there
is also a play consistent with $\sigma$ and starting in $W_0(\game)$ that
contains an unanswered request with cost $\infty$. However, this contradicts
the fact that $\sigma$ is a winning strategy.

Similarly, one can show that $\rho$ has no infix that contains $t+1$ vertices
of some odd color~$c$, but no vertex of a larger color. Using the second
property, a simple induction over the number of odd colors (starting with $1$)
shows that no coordinate storing a score~$s_c$ overflows during $\rho$. Using
this and the first property shows that the second coordinate does not overflow
either. Finally, if the second component does not overflow, then the first
component does not overflow either, since it stores the largest unanswered
request in this case. Hence, the sheets of $\rho$ are strictly smaller than
$\top$. 
\end{proof}

Now we are able to prove our main technical result of this section: using the
score-sheets we can turn an arbitrary uniform finite-state winning strategy into a
positional one. For every $v \in V$, let $P_v$ denote the set of play prefixes
that begin in $W_0(\game)$, are consistent with $\sigma$, and end in $v$. Due
to Lemma~\ref{lemma_fsbounds}, the sheets of the prefixes in $P_v$ are
strictly smaller than $\top$. Hence, for every nonempty $P_v$ there exists a
play prefix $\max_v\in P_v$ such that $\sheet(w) \le \sheet(\max_v) < \top$
for every $w \in P_v$. We define a positional strategy $\sigma'$ by
$\sigma'(wv) = \sigma (\max_v)$.

\begin{lem}
\label{lem_pos}
The strategy~$\sigma'$ is a uniform positional winning strategy for Player~$0$ in
 $\game$. 
\end{lem}

\begin{proof}
An inductive application of Lemma~\ref{lemma_congurence} shows that we have
$\sheet(\rho_0\cdots\rho_n) \le \sheet(\max_{\rho_n})$ for every $n$ and every
play $\rho$ that is consistent with $\sigma'$. Hence, the sheets of $\rho$ are
strictly smaller than $\top$, which implies $\rho \in \bcp(\col)$ due to
Lemma~\ref{lemma_boundedwins}. \end{proof}

In the preliminary version of this work~\cite{FijalkowZimmermann12}, we
presented a similar construction, the main difference being that we did not
use overflows there, but updated a sheet to $\top$ if a full coordinate is
incremented. This construction can also shown to be correct, but the proof of
the analogue of Lemma~\ref{lem_pos} in~\cite{FijalkowZimmermann12} (called
Lemma~15 there) has a gap (the claim in its last line is incorrect). This gap
can be closed using pumping arguments which rely on properties of the bounded
cost-parity condition. Since one of our aims in this section is to give a
general framework that works for other winning conditions as well, we
refrained from presenting the fix and instead changed the definition of the
sheets (adding overflows) to achieve this goal. Indeed, our construction only
relies on the following properties: 
\begin{enumerate} 

\item The score-sheets
constitute a finite total order. 

\item The score-sheet function is a
congruence w.r.t.\ this order. 

\item If the score-sheets of a play are
strictly smaller than the maximal element, then it is winning for Player~$0$.

\item A finite-state winning strategy allows only plays whose score-sheets are
strictly smaller than the maximal element. 

\end{enumerate} 
It follows that for every winning condition for which one can define a scoring
function meeting these conditions, one can turn a finite-state winning
strategy into a positional one. For example, one could extend the sheets
presented above by a new first coordinate that counts how often the second
coordinate (the largest open request) overflows, which corresponds to requests
that are open for \emph{many} increment-edges. Since a finite-state winning
strategy for a parity game with costs bounds the number of such requests, our
framework is applicable to parity games with costs as well. 

On the other hand, our framework cannot applicable to (bounded) Streett games with costs since they are not positionally determined. This impossibility manifests itself in the fact that one cannot totally order the costs of the different requests of a Streett condition while satisfying the other three properties listed above. This is in contrast to (bounded) parity games with costs where larger requests are more important than smaller ones, since answering the larger ones also answers smaller ones. This is reflected in the lexicographic ordering of the sheets. Interestingly, our result above relies on having just one cost function that is used for every request. If we allow different cost functions for different colors, then both players need memory to implement their winning strategies (cf.\ Section~\ref{section_conc}), i.e., our framework cannot be applicable.

Finally, by relaxing the first requirement (the score-sheets being totally ordered) to allow the sheets being partially ordered, then one obtains a finite-state winning strategy whose size is at most the size of the largest anti-chain in the partial order of the sheets (see \cite{NeiderRabinovichZimmermann14} for an application of this idea).

\section{Streett Games with Costs}
\label{section_streett}

In this section, we present an algorithm to solve Streett and bounded Streett
games with costs following the same ideas as in the section about (bounded)
parity games with costs, and prove $\exptime$-completeness of the
corresponding decision problems. From our algorithm, we also obtain upper
bounds on the memory requirements of both players, which are complemented by lower bounds.

The main result of this section is the following theorem. Here, $n$ is the
number of vertices, $m$ is the number of edges, and $d$ is the number of
Streett pairs in the game.

\begin{thm}
\label{theorem_mainstreett}
Given an algorithm that solves Streett games in time $T(n,m,d)$, there is
\begin{enumerate}
\item an algorithm that solves bounded Streett games with costs in
time~$O(T(2^dn, 2^dm, 2d))$.

\item an algorithm that solves Streett games with costs in time~$O(n \cdot 
T(2^dn, 2^dm, 2d))$.
\end{enumerate}
\end{thm}

Fix an arena~$\arena =(V, V_0, V_1, E)$, a collection~$\Gamma = (Q_c, P_c)_{c \in [d]}$ of Streett pairs, and a collection~$\costfam = (\cost_c)_{c \in [d]}$ of cost functions, both compatible with $\arena$. We begin by considering bounded games and again assume for every $c \in [d]$ that no
vertex of $\arena$ has both an incoming increment-edge (w.r.t.\ $\cost_c$) and an incoming
$\epsilon$-edge (again, w.r.t.\ $\cost_c$). Having different types of incoming
edges with respect to different cost functions is allowed. This property can
again be established by subdividing edges. Assuming this, let $I_c$ denote the
vertices with incoming increment-edges w.r.t.\ $\cost_c$. Then, $\coB(I_c) =
\set{\rho \mid \cost_c(\rho) < \infty}$ is the set of plays with finitely many
increment-edges w.r.t.\ $\cost_c$. Let $I = (I_c)_{c \in [d]}$. Furthermore,
we define $\rr(Q_c, P_c)$ to be the set of plays in which every request of pair~$c$ is eventually answered. Finally, we define
\begin{equation*}
 \scrr(\Gamma, I) = \bigcap_{c \in [d]} \big[\, \big( \streett(Q_c, P_c) \cap 
\coB(I_c) \big) \;\cup\; \rr(Q_c, P_c)\, \big]  \enspace,
\end{equation*}
which is $\omega$-regular as a boolean combination of $\omega$-regular languages. This condition is a relaxation of the bounded
cost-Streett condition, as we have $\scrr(\Gamma) \supseteq \bcs(\Gamma)$.

\begin{lem}
\label{lemma_bcs2scrr}
Let $\game = (\arena, \bcs(\Gamma, \costfam))$ and $\game' = (\arena,
\scrr(\Gamma, I))$, where $I$ is defined as above. A winning strategy for
Player~$i$ in $\game'$ from a set~$W$ is also a winning strategy for
Player~$i$ in $\game$ from $W$. Especially, $W_i(\game) = W_i(\game')$ for $i
\in \set{0,1}$.
\end{lem}

The proof of this lemma is similar to the one for Lemma~\ref{lem_bcp2pcrr} and
relies on finite-state determinacy of $\omega$-regular games. 

Next, we show
how to reduce $(\arena, \scrr(\Gamma, I))$ to a classical Streett game: first,
we add a memory structure~$\mem$ of size~$2^d$ that keeps track of the open
requests during a play (cp.~\cite{HornThomasWallmeier08}) and let $F_c$ denote the vertices in which request~$c$
is not open. Then, we have $(\arena, \scrr(\Gamma, I)) \le_\mem (\arena \times
\mem, L)$ with
\[ L =  \bigcap_{c \in [d]} \big[\, \big( \streett(Q_c, P_c) \cap \coB(I_c) 
\big) \;\cup\; \buechi(F_c)\, \big]  \enspace,\]
i.e., we have reduced the request-response conditions~$\rr(Q_c, P_c)$ to
B\"{u}chi conditions\footnote{$\buechi(F_c)$ is the set of plays visiting
$F_c$ infinitely often.}. Finally, we have
\begin{align*}
L =&  \bigcap_{c \in [d]} \big[\, \big( \streett(Q_c, P_c) \cap \coB(I_c) 
\big) \;\cup\; \buechi(F_c)\, \big]\\
=& \bigcap_{c \in [d]} \big[\, \big( \streett(Q_c, P_c) \cap \streett(I_c, \emptyset)  \big) \;\cup\; \buechi(F_c)\, \big]\\
=& \bigcap_{c \in [d]} \big[\, \big( \streett(Q_c, P_c) \;\cup\; \buechi(F_c)\big) \cap \big( \streett(I_c, \emptyset) \;\cup\; \buechi(F_c) \big) \big]\\
=& \bigcap_{c \in [d]} \big[ \streett(Q_c, P_c \cup F_c) \;\cap\; 
\streett(I_c, F_c) \big]  \enspace,
\end{align*} 
which is a Streett condition. Thus, we have reduced $(\arena, \scrr(\Gamma,
I))$ to a Streett game in an arena that is exponential in $d$ with $2d$
Streett pairs. This proves the first claim of
Theorem~\ref{theorem_mainstreett}. Furthermore, we obtain the following upper
bound on the size of finite-state winning strategies. Here, we use the fact
that Player~$0$ has finite-state winning strategies of size $d!$ in Streett
games with $d$ pairs (which is tight), while Player~$1$ has positional winning
strategies~\cite{DziembowskiJurdzinskiWalukiewicz97}. Note that the lower
bound of $d!$ for Player~$0$ is also a lower bound for her in (bounded)
Streett games with costs, since classical Streett games are a special case of
both.

\begin{rem}
\hfill
\begin{enumerate}
\item In bounded Streett games with costs, Player~$0$ has uniform finite-state
winning strategies of size~$2^d ((2d)!)$, where $d$ is the number of Streett
pairs.

\item In bounded Streett games with costs, Player~$1$ has uniform finite-state
winning strategies of size~$2^d$, where $d$ is the number of Streett pairs.
\end{enumerate} \end{rem}

\noindent It remains to show a lower bound on the memory requirements for Player~$1$.

\begin{lem}
For every $d \ge 1$, there is a bounded Streett game with costs~$\game_d$ with a designated vertex~$v$ such 
that
\begin{itemize}

\item the arena of $\game_d$ is of linear size in $d$ and $\game_d$ has $2d$ Streett pairs,

\item Player~$1$ has a uniform finite-state winning strategy for $\game_d$ 
from $v$, which is implemented with $2^d$ memory states, but

\item Player~$1$ has no winning strategy from $v$ that is 
implemented with less than $2^d$ memory states.

\end{itemize}		
\end{lem}

\begin{proof}
The arena~$\arena_d$ of $\game_d$ is depicted in Figure~\ref{fig_streett} where we do not indicate the costs, since every edge is an increment-edge (for every cost function). The winning condition is given as $\Gamma_d = (Q_c, P_c)_{c\in [2d]}$ where $Q_c = \set{q_c}$ and $P_c = \set{s_{c'} \mid c'\neq c}$. The designated vertex~$v$ we consider is $v_0$. Note that every play ends up in one of the sink vertices~$s_c$. 

\begin{figure}
\begin{center}
\begin{tikzpicture}[scale = 1]
\begin{scope}
\tikzset{p0/.style = {shape = ellipse,    draw, thick, minimum size = .6cm}}
\tikzset{p1/.style = {shape = rectangle, minimum size =.6cm, draw, thick}}

\node[p0] 		at (0,0)		(c1)	{$v_0$};
\node[p0] 		at (2,0)		(c2)	{$v_1$};

\node[p0] 		at (1,1)		(q1)	{$q_0$};
\node[p0] 		at (1,-1)		(q2)	{$q_1$};
\node[p0] 		at (3,1)		(q3)	{$q_2$};
\node[p0] 		at (3,-1)		(q4)	{$q_3$};
\node[p0] 		at (5,1)		(q2n-1)	{\tiny $q_{2d-2}$};
\node[p0] 		at (5,-1)		(q2n)	{\tiny $q_{2d-1}$};

\node[p1]		at (6,0)		(c1')	{$v_0'$};
\node[p1]		at (8,0)		(c2')	{$v_1'$};

\node[p0] 		at (7,1)		(p1)	{$p_0$};
\node[p0] 		at (7,-1)		(p2)	{$p_1$};
\node[p0] 		at (9,1)		(p3)	{$p_2$};
\node[p0] 		at (9,-1)		(p4)	{$p_3$};
\node[p0] 		at (11,1)		(p2n-1)	{\tiny $p_{2d-2}$};
\node[p0] 		at (11,-1)		(p2n)	{\tiny $p_{2d-1}$};

\node[p0] 		at (8,2)		(s1)	{$s_0$};
\node[p0] 		at (8,-2)		(s2)	{$s_1$};
\node[p0] 		at (10,2)		(s3)	{$s_2$};
\node[p0] 		at (10,-2)		(s4)	{$s_3$};
\node[p0] 		at (12,2)		(s2n-1)	{\tiny $s_{2d-2}$};
\node[p0] 		at (12,-2)		(s2n)	{\tiny $s_{2d-1}$};

\node			at (3.999,1)	()		{$\cdots$};
\node			at (3.999,-1)	()		{$\cdots$};
\node			at (9.999,1)	()		{$\cdots$};
\node			at (9.999,-1)	()		{$\cdots$};

\end{scope}

\path[thick]
(c1) 		edge				(q1)
(c1) 		edge				(q2)
(q1) 		edge				(c2)
(q2) 		edge				(c2)
(c2) 		edge				(q3)
(c2) 		edge				(q4)
(q2n-1)		edge				(c1')
(q2n)		edge				(c1')
(c1')		edge				(p1)
(c1')		edge				(p2)
(p1) 		edge				(c2')
(p2) 		edge				(c2')
(c2') 		edge				(p3)
(c2') 		edge				(p4)

(p1)		edge				(s1)
(p2)		edge				(s2)
(p3)		edge				(s3)
(p4)		edge				(s4)
(p2n-1)		edge				(s2n-1)
(p2n)		edge				(s2n)

(s1)		edge[loop right]	()
(s2)		edge[loop right]	()
(s3)		edge[loop right]	()
(s4)		edge[loop right]	()
(s2n-1)		edge[loop right]	()
(s2n)		edge[loop right]	()
;

\end{tikzpicture}
\end{center}
\caption{The arena $\arena_d$ for the bounded Streett game with costs~$\game_d$ (every edge is an increment-edge).}
\label{fig_streett}
\end{figure}
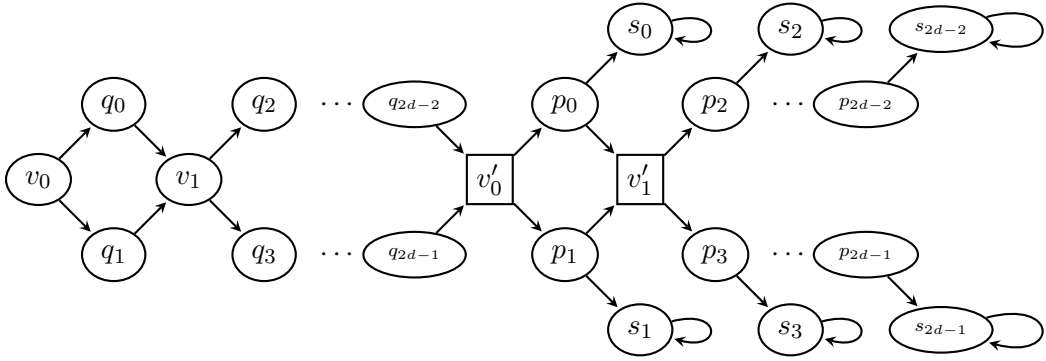

First, we show that Player~$1$ has a winning strategy from $v_0$ that is implemented with $2^d$ memory states. Assume the play is currently ending in vertex~$v_c'$. This play passed through $v_c$ where Player~$0$ either moved to $q_{2c}$ or to $q_{2c+1}$. In the first case, Player~$1$ moves to $p_{2c}$, in the second case to $p_{2c+1}$. Now, consider a play that is consistent with this strategy. It will eventually end up in one of the sink vertices~$s_c$ after visiting the vertex~$p_c$, which implies that the vertex~$q_c$ is visited during the play, too. As $q_c \in Q_c$, a request of condition~$c$ is open, which is never answered, since only sink vertices are in $P_c$, but the sink $s_c$ is not in $P_c$. As every edge is an increment-edge w.r.t.\ $\cost_c$, the play contains an unanswered request of cost~$\infty$, i.e., it is winning for Player~$1$. Note that this strategy can be implemented by memorizing the $d$ binary choices Player~$0$ makes at the vertices~$v_j$, which can be done using $2^d$ memory states.

Now, consider a finite-state strategy~$\tau$ for Player~$1$ implemented with less than $2^d$ memory states. Then, there are two different play prefixes~$w,w'$ leading from $v_0$ to $v_0'$ such that the memory structure reaches the same memory state after processing these two prefixes. Since these prefixes differ, there is a $c$ such that $w$ visits $q_{2c}$ and $w'$ visits $q_{2c+1}$. Now, consider the prolongations of these prefixes where Player~$1$ plays according to $\tau$ and Player~$0$ does not move to the sink vertices in order to reach vertex~$v_{c}'$. Since the memory structure reaches the same memory state after processing $w$ and $w'$, it behaves the same after processing these prolongations. Hence, $\tau$ makes the same move after both prolongations, say it moves to $p_{2c}$ (the case $p_{2c+1}$ is analogous). Now consider the prolongation of $w'$: Player~$1$ moves to $p_{2c}$ and then Player~$0$ can move to the sink vertex~$s_{2c}$, where the requests of every condition but $2c$ are answered to. However, a request of condition~$2c$ is not open during this play, since $w'$ visited $q_{2c+1}$ and therefore did not visit $q_{2c}$. Hence, every request is answered and no new ones are raised in the sink. Hence, the play is winning for Player~$0$. Thus, the strategy~$\tau$ cannot be winning for Player~$1$.
\end{proof}
Note that the game $\game_d$ presented above is even a bounded Streett game~\cite{ChatterjeeHenzingerHorn09}.

Now, we consider Streett games with costs: we again show that solving the bounded
variant suffices to solve such games.

\begin{lem}
\label{lemma_csfixpoint}
Let $\game = (\arena, \cs(\Gamma, \costfam))$ and let 
$\game' = (\arena, \bcs(\Gamma, \costfam))$.
\begin{enumerate}

\item\label{lemma_csfixpoint_playerzerosubset} $W_0(\game') \subseteq 
W_0(\game)$.

\item\label{lemma_csfixpoint_playerzeroempty} If $W_0(\game') = \emptyset$, 
then  $W_0(\game) = \emptyset$.
\end{enumerate}
\end{lem}

The proof is exactly the same as the one for Lemma~\ref{lem_cp2bcp_fixpoint}.
Also, Algorithm~\ref{algorithm_fixpoint} (where $X_j$ is now Player~$0$'s
winning region in the bounded Streett game with costs) works for this pair of
winning conditions as well. This proves the second claim of
Theorem~\ref{theorem_mainstreett}.

Furthermore, using the same construction as presented in the proof of
Lemma~\ref{lem_fixpointalgo_correctness}, one can built a winning strategy for
a Streett game with costs out of the winning strategies for the bounded
Streett games with costs solved by (the modified)
Algorithm~\ref{algorithm_fixpoint}. By reusing memory states in the different sets~$X_j$ computed by the algorithm, we obtain the
following upper bound for Player~$0$. Note that we can reuse memory states, since no information needs to be transferred between the regions~$X_j$: once a set $X_j$ is entered, the strategy forgets about the history of the play. 

\begin{rem}
In Streett games with costs, Player~$0$ has uniform finite-state winning 
strategies of size~$2^d ((2d)!)$, where $d$ is the number of Streett pairs.
\end{rem}

Again, the lower bound of $d!$ for Player~$0$ in classical Streett games is also a lower bound for her in Streett games with costs. Player~$1$ on the other hand needs infinite memory in Streett games with 
costs, as witnessed by the game in Example~\ref{example_costparity}, which can 
be easily transformed into a Streett game with costs.

Using the algorithm presented in~\cite{PitermanPnueli06}, which solves a
Streett game in time~$O(mn^ddd!)$, one can solve (bounded) Streett games with
costs in exponential time, although the Streett games that need to be solved
are of exponential size (but only in $d$). Here it is crucial that the number of Streett pairs only grows linearly. Together with the
$\exptime$-hardness of solving bounded and finitary Streett
games\footnote{Shown in unpublished work by Chatterjee, Henzinger, and Horn,
obtained by slightly modifying the proof of $\exptime$-hardness of solving
request-response games~\cite{ChatterjeeHenzingerHorn11}.}, which are a special
case of (bounded) Streett games with costs, we obtain the following result.

\begin{thm} 
The following problem is $\exptime$-complete: Given a (bounded) Streett game with costs $\game$,
$i \in \set{0,1}$, and a vertex~$v$, is $v \in W_i(\game)$? 
\end{thm}

\section{Conclusion}
\label{section_conc}

We introduced infinite games with cost conditions, generalizing both classical
conditions and finitary conditions. For parity games with costs, we proved
half-positional determinacy and that solving these games is not harder than
solving parity games. The decision problem is in
$\np \cap \conp$ (this was recently improved to $\up \cap \coup$~\cite{DBLP:conf/lpar/MogaveroMS13}).

 For Streett games with costs, we showed that Player~$0$ has
finite-state winning strategies and that solving these games is not harder
than solving finitary Streett games and can be done by solving linearly many
(classical) Streett games of exponential size (in the number of Streett
pairs). Table~\ref{table_results} sums up all our results on games with costs and compares them
to the results for the classical and finitary variants. Here, $d$ denotes the number of odd colors in the game and
``exponential'' is always meant to be ``exponential in the number of
Streett-pairs''. The memory bounds for the different types of parity games are
tight, while there are gaps between the exponential lower and the exponential
upper bounds for the different types of Streett games with costs. 

\begin{table*}[h]
\centering
\begin{tabular}{llll}
\toprule

winning condition\phantom{n} & computational complexity\phantom{n} & memory
Pl.~$0$\phantom{n} & memory Pl.~$1$\\

\midrule

parity				& $\up \cap \coup$ 	& positional 	& positional\\
bounded parity 		& $\ptime$ 			& positional 	& $d+1$\\
finitary parity		& $\ptime$ 			& positional 	& infinite\\
bounded cost-parity	& $\up \cap \coup$	& positional 	& $d+1$\\
cost-parity			& $\up \cap \coup$ 	& positional 	& infinite\\

\midrule

Streett 			& $\conp$-complete 		& exponential	& positional\\
bounded Streett		& $\exptime$-complete	& exponential	& exponential\\
finitary Streett 	& $\exptime$-complete 	& exponential	& infinite\\
bounded cost-Streett& $\exptime$-complete 	& exponential	& exponential\\	
cost-Streett 		& $\exptime$-complete 	& exponential	& infinite\\

\bottomrule
\end{tabular}
\caption{Overview of computational complexity and memory requirements.}
\label{table_results}
\end{table*}

Let us discuss two variations of the games presented here. In a parity game
with costs, the requests and responses are hierarchical and there is a single
cost function that is used for every request. On the other hand, in Streett
games with costs, the requests and responses are independent and there is a
cost function for every pair of requests and responses. Thus, there are two
other possible combinations.

First, consider parity games with multiple cost functions (one for each odd
color): a reduction from QBF shows that solving such games is $\pspace$-hard.
On the other hand, the problem is in $\exptime$, since every such game is a
Streett game with costs. Furthermore, one can show that Player~$0$ needs
exponential memory (in the number of odd colors) to implement her winning
strategies. All these results even hold for the bounded variant of these game,
which is defined as one would expect. In these games, both players need
exponential memory. In further research we aim at closing the gap in
complexity of solving parity games with multiple cost functions.
The second variation are Streett games with a single cost function. Solving
finitary Streett games is already $\exptime$-complete and our lower bounds on
memory requirements are derived from Streett games. Note that both finitary
Streett and classical Streett games can be seen as Streett games with a single
cost function. Hence, these games are as hard as Streett games with multiple cost functions.

Finally, there are at least two other directions to extend our results
presented here: first, our winning conditions do not cover all acceptance
conditions (for automata) discussed in~\cite{BojanczykColcombet06, Boom11}. In
ongoing research, we investigate whether our techniques are applicable to
these more expressive conditions and to winning conditions specified in
weak MSO with the unbounding quantifier~\cite{Bojanczyk11,
BojanczykTorunczyk12}. Finally, one could add decrement-edges.

\bibliographystyle{plain}
\bibliography{bibliography}

\end{document}